\documentclass[journal]{IEEEtran} % ,10pt,onecolumn
\IEEEoverridecommandlockouts 

\usepackage[noadjust]{cite}
\usepackage{mathtools}
\usepackage{amsmath,amssymb,mathrsfs,amsfonts,dsfont,amsthm,soul}
\usepackage{enumerate}
\usepackage[english]{babel}
\usepackage{comment}
\usepackage{epsfig,amsfonts,latexsym,amssymb,shadow,amsmath,wrapfig}%,eepicemu}%,psfrag}
\usepackage[usenames]{color}
\usepackage{tikz}
\usepackage{algorithm}
\usepackage[noend]{algpseudocode}
\usepackage{footnote}
\usepackage{tablefootnote}
\makesavenoteenv{tabular}

\newtheorem{theorem}{Theorem}

\newtheorem{lemma}{Lemma}

\newtheorem{remark}{Remark}

\newtheorem{assum}{Assumption}

\DeclareMathOperator*{\argmin}{\arg\!\min}
\DeclareMathOperator*{\argmax}{\arg\!\max}
\newcommand{\minitab}{\hspace*{2em}}
\begin{document}
%
% paper title
% Titles are generally capitalized except for words such as a, an, and, as,
% at, but, by, for, in, nor, of, on, or, the, to and up, which are usually
% not capitalized unless they are the first or last word of the title.
% Linebreaks \\ can be used within to get better formatting as desired.
% Do not put math or special symbols in the title.
\title{Optimal Auction Design for Flexible Consumers}
%
%
% author names and IEEE memberships
% note positions of commas and nonbreaking spaces ( ~ ) LaTeX will not break
% a structure at a ~ so this keeps an author's name from being broken across
% two lines.
% use \thanks{} to gain access to the first footnote area
% a separate \thanks must be used for each paragraph as LaTeX2e's \thanks
% was not built to handle multiple paragraphs
%

\author{Shiva~Navabi,~\IEEEmembership{Student Member,~IEEE,}
        Ashutosh~Nayyar,~\IEEEmembership{Member,~IEEE}
%        and~Jane~Doe,~\IEEEmembership{Life~Fellow,~IEEE}% <-this % stops a space
\thanks{S. Navabi and A. Nayyar are with the Electrical Engineering Department, University of Southern California, 3740 McClintock Avenue, Los Angeles, CA, 90089 USA (E-mail: navabiso@usc.edu; ashutosn@usc.edu.)}% <-this % stops a space
%\thanks{A. Nayyar was with the Electrical Engineering Department, University of Southern California, Los Angeles,
%CA, 90089 USA e-mail: ashutosn@usc.edu.}
%\thanks{J. Doe and J. Doe are with Anonymous University.}% <-this % stops a space
%\thanks{Manuscript received June 15, 2016; revised August 15, 2016.}
}

% make the title area
\maketitle

% As a general rule, do not put math, special symbols or citations
% in the abstract or keywords.
\begin{abstract}
%\textcolor{blue}{TO Be Edited}

We study the problem of designing revenue-maximizing auctions for allocating multiple goods to flexible consumers. In our model, each consumer is interested in a subset of goods known as its flexibility set and wants to consume one good from this set. A consumer's flexibility set and its utility from consuming a good from its flexibility set are its private information.  We focus on the case of nested flexibility sets --- each consumer's flexibility set can be one of   $k$ nested sets. We provide several examples where such nested flexibility sets may arise. We characterize the allocation rule for an incentive compatible, individually rational and revenue-maximizing auction in terms of  solutions to  integer programs. The corresponding payment rule is described by an integral equation. We then leverage the nestedness of flexibility sets to simplify the optimal auction and provide a complete characterization of allocations and payments in terms of simple thresholds.

%The problem of designing Bayesian incentive compatible, individually rational and revenue maximizing auction for multiple goods and flexible customers is considered. The auctioneer has $M$ goods  and $N$ potential customers.  Customer $i$ has a flexibility set $\phi_i $ which represents the set of goods the customer is equally interested in. Customer $i$ can consume at most one good from its flexibility set. We first characterize the optimal auction for customers with arbitrary flexibility sets and then consider the  case when the flexibility sets are nested. This allows us to group customers into  classes of increasing flexibility.  We show that the optimal auction can be simplified in this case and we provide a complete characterization of allocations and payments in terms of simple thresholds.
% The problem is addressed for the case where the inventory of goods is fixed as well as one in which additional supplies can be purchased. 
\end{abstract}

% Note that keywords are not normally used for peerreview papers.
\begin{IEEEkeywords}
Revenue maximization, Bayesian incentive compatibility, flexible demand, optimal auction.
\end{IEEEkeywords}

% For peer review papers, you can put extra information on the cover
% page as needed:
% \ifCLASSOPTIONpeerreview
% \begin{center} \bfseries EDICS Category: 3-BBND \end{center}
% \fi
%
% For peerreview papers, this IEEEtran command inserts a page break and
% creates the second title. It will be ignored for other modes.
\IEEEpeerreviewmaketitle

\section{Introduction}\label{sec:intro}

The problem of allocating limited resources among multiple users arises frequently in a wide array of applications ranging from communication networks to transportation and power systems.  In many such applications, the users are selfish agents with private information about their preferences and constraints.  Finding a desirable allocation of resources would typically require at least a partial knowledge of  users' private preferences and constraints.  The users, however,  can behave strategically in revealing their private information to benefit themselves at the expense of other users and/or the owner of the resources being allocated.  Thus, the presence of strategic users with private information creates two key challenges for the resource allocation problem: (i) the allocation needs to be based on the information revealed by the users; (ii) the allocation procedure must anticipate users' strategic behavior in the revelation of their private information. 
 The economic theory of mechanism design provides a framework for addressing such resource allocation problems.
 
 Auctions provide one of the simplest settings of a mechanism design problem. 
 An auctioneer/mechanism designer would typically ask for bids from potential customers, and allocate resources and charge payments as a function of the received bids. Customers with private information about their utilities can be strategic about what bids they submit. The auction design problem is to find suitable allocation and payment functions, which map the customers' bids to allocations and payments, so that the auctioneer can achieve some desired objective. Typically, the auctioneer's objectives are either maximization of its revenue or maximization of social welfare.
 
 In this paper, we consider the problem of designing revenue-maximizing auctions for multiple goods and flexible consumers. Consumer flexibility about goods can arise in different scenarios. In demand  response programs of electric utilities, some consumers may be flexible about when and at what rate they receive power. In airline/hotel reservation settings, customers may be flexible about their travel dates. The seller of these goods/services should be able to take this flexibility into account to improve its profits.  In our setup, each consumer is associated with a  \emph{flexibility set} that describes the subset of goods the consumer is equally interested in. Each consumer wants to consume \emph{one good from its flexibility set}. The flexibility set of a consumer  and the utility it gets from consuming a  good from its flexibility set are both its private information. 
 
 We focus on the case of nested flexibility sets --- each consumer's flexibility set can be one of the  $k$ sets, $\mathcal{B}_1, \mathcal{B}_2,\ldots, \mathcal{B}_k$, which are nested in the following way:
\begin{equation}\label{Bsets}
\mathcal{B}_1 \subset \mathcal{B}_2 \subset \cdots \subset \mathcal{B}_{k}.
\end{equation}
If consumers' flexibility  sets are truthfully revealed to the auctioneer, the nestedness in \eqref{Bsets} allows the auctioneer to  compare consumer flexibility and say whether a given consumer is more, less or equally flexible as another consumer.
\vspace{-1em}
\subsection{Examples of Nested Flexibility}
There are several markets where consumer flexibility resembles the nested pattern in \eqref{Bsets}. For example,  consider flexible electricity consumers that need one unit of energy within a certain deadline \cite{7548363}. Let $\mathcal{B}_{\tau}$ denote the set of energy units available for delivery in the interval $[0,\tau]$, $\tau=1\ldots,k$. Clearly, $\mathcal{B}_1 \subset \mathcal{B}_2 \subset \cdots \subset \mathcal{B}_k$.  A consumer who needs one unit of energy with a deadline of  $2$  can be seen as having $\mathcal{B}_2$ as its flexibility set, that is, it needs one good from $\mathcal{B}_2$. A similar flexibility model appears in auctions with  deadline-based goods such as airline ticket auctions where different customers may have different departure deadlines.

  As another example,  consider electricity consumers that need to receive a fixed amount of energy within a fixed time interval while having certain constraints on the  rate at which they can receive energy. Suppose each consumer needs to receive one unit of energy within the time interval $[0,T]$ but some consumers need energy at a constant rate while others can tolerate variable rates. Let $\mathcal{B}_1$ be the set of energy units that the energy provider can supply at a constant rate over the interval $[0,T]$ and $\mathcal{B}_2$ be the set of all  energy units that can be supplied over the interval $[0,T]$. We thus have consumers whose flexibility sets  are either $\mathcal{B}_1$ or $\mathcal{B}_2$ with $\mathcal{B}_1 \subset \mathcal{B}_2$.

%and let $R_i$ denote the maximum rate that is allowed for user $i$'s device to receive energy at. Let $\mathcal{B}_i, i = 1, 2, \cdots, k$ denote the set of energy units that the grid can supply within the time interval $[0,T]$ with delivery rate no larger than $q_i$ such that $q_1 < q_2 < \cdots < q_k$. Hence, the sets $\mathcal{B}_i, i = 1, 2, \cdots, k$ exhibit the nested pattern: $\mathcal{B}_{1} \subset \mathcal{B}_{2} \subset \cdots \subset \mathcal{B}_{k}$.

Another example of consumer flexibility comes from auction-based  spectrum allocation in cognitive radio networks (\cite{khaledi2013auction}, \cite{sengupta2008designing}, \cite{zhang2012auction}) where a primary spectrum owner has multiple frequency bands with different bandwidths. These bands  can be allocated to secondary users who need a certain minimum amount of bandwidth.  Suppose the primary owner has frequency bands of widths $w_1, w_2, \cdots, w_{k}$ with  $w_1 < w_2 < \cdots < w_k$. Let $\mathcal{W}_i, i = 1, 2, \cdots, k,$ denote the set of frequency bands of width $w_i$ that are available for allocation to secondary users.  Define $\mathcal{B}_i = \bigcup\limits_{j=k-i+1}^k \mathcal{W}_j, i = 1, 2, \cdots, k,$ as the set of frequency bands of width greater than or equal to $w_{k-i+1}$.  We thus have $\mathcal{B}_{1} \subset \mathcal{B}_{2} \subset \cdots \subset \mathcal{B}_{k}$. A secondary user that needs one frequency band of width at least $w_i$ can be interpreted as having  $\mathcal{B}_{k-i+1}$ as its flexibility set.

Consider next auction-based content delivery in Wireless Information Centric Networks \cite{mangili2016bandwidth} where multiple content providers compete for limited cache storage resources provided by a Wireless Access Point (WAP) in a given region for a certain time period. Suppose the WAP has $k$ cache servers with storage capacities $c_1 < c_2 <\ldots < c_k$. Assume that one cache server can serve at most one content provider at a time. Let $\mathcal{B}_i$ be the set of cache servers with capacity greater than or equal to $c_{k-i+1}$. Clearly the sets $\mathcal{B}_i, i=1,\ldots,k,$ are nested. A content provider who needs a  cache of storage capacity   at least $c_{k-i+1}$ has the flexibility set $\mathcal{B}_i$.

%Assuming that each cache server can serve one content provider at a time, the users can be grouped according to their storage capacity demand. Let $C_i$ denote the storage capacity that user $i$ needs and and let $c_1, c_2, \cdots, c_k$ denote the storage capacities of the cache servers such that $c_1 < c_2 < \cdots c_k$. Let $\mathcal{D}_i, i = 1, 2, \cdots, k$ denote the set of cache servers of capacity $c_i$ available for allocation and define $\mathcal{B}_i = \bigcup\limits_{j=i}^k \mathcal{D}_j, i = 1, 2, \cdots, k$ as the set of cache servers of capacity greater than or equal to $c_i$; hence, it is observed that: $\mathcal{B}_{k} \subset \mathcal{B}_{k-1} \subset \cdots \subset \mathcal{B}_{1}$.

%Another example of a similar market demand is in flight tickets booking auctions where the airline needs to serve the demands of a set of passengers who need to travel from city A to city B within a restricted time period. Suppose each buyer demands 1 airplane seat; thus, customers can be grouped according to their departure deadlines. Define $\mathcal{B}_{\tau}\; , \; \tau = 1, 2, \cdots, k$ as the set of flight tickets available for allocation in the time interval $[0,\tau]$ and hence, these sets follow the nested pattern as in $\mathcal{B}_1 \subset \mathcal{B}_2 \subset \cdots \subset \mathcal{B}_k$.
\vspace{-1em}
\subsection{Comparison with Prior Literature}\label{sec:prior-lit}
The problem of designing auctions has been investigated under many different setups in the prior literature and can be broadly categorized on the basis of (a) the problem objective (revenue or social welfare maximization), (b) the nature of supply (single unit or multiple units, identical or non-identical goods), (c) the nature of demand (unit demand, demand for bundles, etc), and (d) the nature of private information (one-dimensional or multi-dimensional).

%\subsubsection{Efficient Auctions}  \label{sec:efficient}
%The problem of designing auctions has been investigated in many different setups in the prior literature.
Numerous works have addressed social welfare maximizing or \textit{efficient} auctions, the most well-known of these being the Vickrey-Clarke-Groves (VCG) mechanism \cite{vickrey1961counterspeculation}, \cite{clarke1971multipart}, \cite{groves1973incentives}. Efficient auctions have also been extensively studied in the context of combinatorial auctions (\cite{cramton2006combinatorial}, \cite[Chapter 8]{young2014handbook}, \cite[Chapter 11]{nisan2007algorithmic}). Under some scenarios the problem of exactly maximizing social welfare may not be tractable. Some works such as \cite{jin2015quality} and \cite{babaioff2009single} have thus focused on \textit{approximate} social welfare maximizing auctions. Our focus in this paper, however, is on   revenue-maximizing auctions. In the context of revenue-maximizing auctions, we can categorize the relevant literature as follows: \\

\emph{1. Multi-unit auctions with identical goods:} 
This strand of literature has focused on  revenue-maximizing auctions in settings where the seller has a number of identical goods and wants to allocate them among several consumers who may demand one or multiple units. In his seminal paper \cite{myerson1981optimal}, Myerson derived fundamental results for the single-unit revenue maximizing auction. In sequel, several works studied revenue-maximizing multi-unit auctions with identical goods under various assumptions about the consumers' utility functional forms and private information structure.  The setups in  \cite{harris1981theory} and  \cite{maskin1989optimal}, for instance, include the problem of auctioning multiple identical goods among consumers with unit demand and private valuations. \cite{malakhov2009optimal} considered the auction of multiple identical goods to  consumers with limited capacities for the number of goods they can consume.
% A consumer's valuation and capacity  are its private information and take values in discrete spaces. %They consider a discrete type space and formulate a linear program that can be viewed as a parametric shortest path problem on a lattice. 

A key feature of these models is that all goods are perceived to be identical by all consumers. Thus,  consumers care only about the number of goods they receive and not about the identities of the goods received. In contrast, consumers in our model differentiate between goods according to their flexibility sets. For example, a consumer with flexibility set $\mathcal{B}_1$ differentiates between goods in $\mathcal{B}_1$ (which give it a positive utility) and goods not in $\mathcal{B}_1$ (which give no utility) whereas a  consumer with a different flexibility set would view goods differently. In other words,   the distinction between  goods is  made subjectively by each consumer based on  its flexibility set.  

~\\
\emph{2. Combinatorial Auctions:} The problem of designing revenue-maximizing auctions has also been investigated in the context of combinatorial auctions \cite[Section 5.2]{de2003combinatorial}.
When the seller has multiple heterogeneous items to auction, consumers may have different utilities for different {subsets} of items due to complementarities and substitution effects. Combinatorial auctions   provide a framework where  consumers can place bids on various \textit{combinations/bundles} of  goods. Some key setups explored under this umbrella are:  \\

\emph{2.1 Auctions with two non-identical goods:} Armstrong \cite{armstrong2000optimal} studied  revenue-maximizing auction for the case where the seller wants to sell two non-identical goods to several consumers. Each consumer can receive one or both of the goods and has a pair of valuations, one for each of the two goods. The valuations are drawn from binary sets and are  independent across the consumers. Avery et al. \cite{avery2000bundling} considered a similar setup as in Armstrong \cite{armstrong2000optimal} with a   single identifiable consumer who may wish to buy both objects and a number of other consumers who wish to buy only one or the other of the two objects.  Two key features that differentiate these setups from our model are: (a)  in both these setups, one or more consumers can consume more than one good whereas in our model each consumer can consume at most one good, (b) the number of goods in our model is not restricted to be two. \\

\emph{2.2 Auctions with single-minded consumers:} Some recent works have considered an extreme case of complementarity among  goods in multi-unit auctions by imposing the assumption of having single-minded consumers.  A single-minded consumer is interested in getting \textit{all} goods from a certain subset of goods. This is in clear contrast to our setup where each consumer wants to get \textit{one} good from its flexibility set. 

Ledyard \cite{ledyard2007optimal} characterized a revenue-maximizing dominant strategy  auction for single-minded consumers where each consumer's desired bundle is known to the seller and a  consumer's valuation constitutes its  one-dimensional private information. Unlike the model in \cite{ledyard2007optimal},  both valuation and flexibility set are a consumer's private information in our model.
Abhishek and Hajek \cite{abhishek2010revenue} considered optimal auction design for single-minded consumers with each user's preferred bundle as well as its valuation being its private information. The single-minded nature of the consumers differentiates this work from our  flexible consumer model. \\

\emph{2.3 Auctions with one-dimensional private information:} \cite{hartline2007profit} surveys   revenue-optimal auctions in various settings where each consumer's private information comprises only its valuation  and is  assumed to be one-dimensional.  \cite{levin1997optimal} studied an optimal auction design problem   for two non-identical goods where each consumer's valuation function is parametrized by a single quantity that represents its one-dimensional private information. Consumers' private information in our model, however,  is two-dimensional, consisting of both valuation and flexibility set. \\

%As a result, complementarities and substitution effects between the offered goods arise that can be leveraged to develop more profitable allocations.

\emph{2.4 A general combinatorial auction:} \cite[Section 5.2]{de2003combinatorial} considered a general setup where the seller has multiple distinct goods and each consumer has a value function that describes its valuation for each bundle of goods. For each allocation rule,  \cite{de2003combinatorial} provides an optimization problem (in fact, a linear program) whose solution (if it exists) gives a payment   rule  that satisfies incentive compatibility and individual rationality constraints. Eventually, after  linear programs corresponding to all possible choices of the allocation rule are solved, the   allocation and payment rules that yield the highest revenue are declared as the revenue-maximizing mechanism. As pointed out in  \cite[Section 5.2]{de2003combinatorial}, this approach is computationally very demanding because the number of possible allocation rules can be very large and no closed-form solutions are available  in general. 

Our model can be seen as a special case of the general framework of  \cite{de2003combinatorial}.  A consumer with flexibility set $\mathcal{B}_i$  and valuation $\alpha$ can be viewed as having a value function of the form:
%\begin{align}\label{eq:valueF}
%v_i(\mathcal{S}) = \left\{
%    \begin{array}{ll}
%        \alpha & \mbox{if } \; \mathcal{S} \cap \mathcal{B}_i \neq \emptyset \\
%        0 & \mbox{if } \; \mathcal{S} \cap \mathcal{B}_i = \emptyset
%    \end{array}
%\right..
%\end{align}
\vspace{-.4em}
\begin{align}\label{eq:valueF}
v_i(\mathcal{S}) = \left\{
    \begin{array}{ll}
        \alpha & \mbox{if } \; \mathcal{S} \cap \mathcal{B}_i \neq \emptyset \mbox{~and~} |\mathcal{S}|=1 \\
        0 & \mbox{otherwise } \; 
    \end{array}
\right..
\end{align}
In section \ref{sec:NestedPhis}, we show that under the assumption of nested flexibility sets,  we can find the optimal auction in a much more straightforward and computationally simpler way than the one described in \cite{de2003combinatorial}. In particular, unlike the case in \cite[Section 5.2]{de2003combinatorial}, we do not need to solve a separate  optimization problem for every possible allocation rule which results in a significant reduction in the computational cost. It should also be noted that a consumer's private information (its value function) in \cite{de2003combinatorial} is drawn from a finite set whereas the valuation in our model is a continuous variable.

\begin{remark}
Consumers in our model want one good from their flexibility sets. This model can be viewed as a special case of the models in \cite{demange1986multi} and  \cite{ashlagi2009ascending} where a number of (potentially non-identical) goods are to be allocated among several consumers and each consumer is interested in receiving at most one good. Unlike our objective of revenue-maximization, the objective in \cite{demange1986multi} and  \cite{ashlagi2009ascending} is to find minimal competitive  prices and equilibrium assignments to clear the market.
\end{remark}

\begin{remark}
The model and results in \cite{7548363} are fundamentally different from those in our paper. In particular,  \cite{7548363} deals with a continuum of consumers. This is crucial because it implies that a single consumer cannot influence the ``aggregate demand bundle'' (as defined in \cite{7548363}) and hence the prices. This is in stark contrast to our paper (and most auction design problems) with finitely many consumers where each consumer can influence the prices through its reports/bids. This means that each individual consumer can strategically manipulate its report to influence allocation and prices in our problem whereas it has no effect on prices in \cite{7548363}. This, we believe, makes our paper conceptually very different from \cite{7548363}.  
\end{remark}

\vspace{-1em}
\subsection{Organization} \label{sec:ORG}
The rest of the paper is organized as follows: we discuss the problem formulation and the mechanism setup in Section \ref{sec:Formul}. In Section \ref{sec:BIC-IR-Mechs}, we characterize incentive compatibility and individual rationality constraints for the mechanism.  We show that the optimal allocation is the solution to an integer program in Section \ref{sec:RMM}. In Section \ref{sec:NestedPhis}, we   simplify the optimal allocation and payments and characterize them in terms of simple thresholds. We summarize our findings and briefly point out potential extensions to the current framework in Section \ref{sec:recap} . 
% The problem of revenue maximization with added supply is considered in section \ref{sec:gdecision}. 
\vspace{-1em}
\subsection{Notations}
$\{ 0,1\}^{N \times M}$ denotes the space of $N\times M$ dimensional matrices with entries that are either 0 or 1. $\mathbb{Z}^+$ is the set of non-negative integers. For a set $\mathcal{A}$, $|\mathcal{A}|$ denotes the cardinality of $\mathcal{A}$.  $x^+$ is the positive part of the real number $x$, that is, $x^+ = \max(x,0)$. 
Vector inequalities are component-wise; that is, for two $1\times n$ dimensional vectors $\textbf{u}=(u_1, \cdots, u_n)$ and $\textbf{v}=(v_1, \cdots, v_n)$, $\textbf{u} \le \textbf{v}$ implies that $u_i \le v_i \; , \text{for} \; i=1, \cdots, n$. The transpose of a vector $\textbf{u}$ is denoted by $\textbf{u}^T$.   $\mathds{1}_{\{a\le b\}}$ denotes 1 if the inequality in the subscript is true and 0 otherwise. $\mathbb{E}$ denotes the expectation operator. For a random variable/random vector $\theta$, $\mathbb{E}_{\theta}$ denotes that the expectation is with respect to the probability distribution of $\theta$. 
%\textcolor{blue}{Notation $\{ 0,1\}^{N \times M}$, cardinality of sets, set of positive integers, $x^+ = max(x,0)$, vector inequalities are component-wise, $\mathds{1}$ is the indicator function}
\vspace{-.8em}
\section{Problem Formulation}  \label{sec:Formul}
We consider a setup where an auctioneer has $M$ goods and $N$ potential customers. $\mathcal{M} = \{1, 2, \cdots, M\}$ denotes the set of goods and  $\mathcal{N} = \{1, 2, \cdots, N\}$ denotes the set of potential customers. Customer $i$, $i \in \mathcal{N}$, has a flexibility set $\phi_i \subset \mathcal{M}$ which represents the set of goods the customer is equally interested in.  Customer $i$ can consume at most one good from its flexibility set $\phi_i$.  We assume  that the flexibility set of each customer can be one of $k$ nested sets. That is, we have $k$ nested subsets of the set of goods: 
\begin{equation} \label{FlexSets_nest}
\mathcal{B}_1 \subset \mathcal{B}_2 \subset \cdots \subset \mathcal{B}_{k} \subseteq \mathcal{M},
\end{equation}
and  $\phi_i \in \{\mathcal{B}_1 , \mathcal{B}_2, \cdots, \mathcal{B}_k\}$ for every $i \in \mathcal{N}$.  If $\phi_i=\mathcal{B}_j$, we  say that customer $i$'s \emph{flexibility level}, denoted by $b_i$, is $j$. Customer $i$'s utility from receiving a good from $\phi_i$ is $\theta_i$.

 We assume that $\theta_i$ and $b_i$ are customer $i$'s private information and are unknown to other users as well as the auctioneer. We assume that  $(\theta_i,b_i), i \in \mathcal{N},$ are independent random pairs taking values in the product sets $ [\theta_i^{min}, \theta_i^{max}] \times \{1,2,\ldots,k\}$\footnote{$\theta_i^{min}$ is assumed to be non-negative.}, $i \in \mathcal{N}$, respectively. The  probability distributions\footnote{We assume that $f_i(\theta_i,b_i) > 0,$ for all $(\theta_i,b_i) \in [\theta_i^{min}, \theta_i^{max}] \times \{1,2,\ldots,k\}, \forall i \in \mathcal{N}$.} $f_i$ of $(\theta_i,b_i)$, $i \in \mathcal{N},$ are assumed to be common knowledge. We define $\theta \coloneqq (\theta_1, \theta_2, \cdots, \theta_N)$ and $b \coloneqq (b_1, b_2, \cdots, b_N)$ as the customers' valuations profile and flexibility levels profile, respectively. $f(\theta,b)$ is the joint probability distribution of $(\theta,b)$. Let $\Theta_i:= [\theta_i^{min}, \theta_i^{max}]$ and  $\Theta := \prod\limits_{i=1}^N \Theta_i$. The pair 
 $(\theta_i,b_i)$ is referred to as customer $i$'s type.

An allocation of the goods among the customers can be described by  an $N \times M$ dimensional matrix $\mathbf{A}$ with the entry $\bold{A}(i,j) = 1$ if customer $i$ gets good $j$ and $\bold{A}(i,j) = 0$ otherwise.
%$[a_{ij}]_{\substack{{i=1,\cdots,N}\\{j=1,\cdots,M}}}$ such that
%\begin{equation} \label{A}
%\bold{A}(i,j) = \left\{
%    \begin{array}{ll}
%        1 & \text{if customer $i$ gets good $j$} \\
%        0 & \text{otherwise}  
%    \end{array}
%\right.
%\end{equation}
The matrix $\mathbf{A}$ is called an allocation matrix.  We assume that the goods are indexed such that the first $|\mathcal{B}_l|$ goods belong to $\mathcal{B}_l$, for $l=1,\ldots,k$.

We require that each of the $M$ available goods  be allocated to {at most} one customer and that each customer receives at most one good. This implies that  $\sum\limits_{i=1}^N \textbf{A}(i,j) \le 1\; , \forall j$ and $\sum\limits_{j=1}^M \textbf{A}(i, j) \le 1\; ,\forall i$. 
A binary matrix $\textbf{A} $ that satisfies these two constraints is called a \textit{feasible} allocation matrix. Let $\mathcal{S} \subset \{0,1\}^{N\times M}$ denote the set of all feasible allocation matrices. That is, \vspace{-0.4cm}
\begin{equation} \label{S}
\begin{split}
\mathcal{S} \coloneqq \Big\{ &\textbf{A} \in \{0,1\}^{N \times M} \: \mid \: \sum\limits_{i=1}^N \textbf{A}(i,j) \le 1 \: , \\ &\forall j \in \mathcal{M} \; , \;
\sum\limits_{j=1}^M \textbf{A}(i,j) \le 1 \; , \; \forall i \in \mathcal{N} \: \Big\}.
\end{split}
\end{equation}
Given an allocation matrix $\mathbf{A}$ and a payment $t_i$ charged to customer $i$, the net utility for this customer is
\begin{equation} \label{ui}
u_i(\theta_i, b_i, \textbf{A}, t_i) = \theta_i \Big( \sum\limits_{j \in \mathcal{B}_{b_i}} \textbf{A}(i, j) \Big) - t_i.
\end{equation}
%Where $t_i$ is the payment user $i$ is charged in the event she is allocated some good from her flexibility set $\phi_i$. The summation $\sum\limits_{j \in \phi_i} \textbf{A}(i, j)$ equals the number of goods user $i$ is allocated from $\phi_i$. 
\vspace{-3em}
\subsection{The Mechanism}  \label{sec:MechSetup}
We consider  direct  mechanisms  where, for each $i \in \mathcal{N}$, customer $i$  reports a valuation from the set $\Theta_i$ and a flexibility level from the set $\{1, 2, \cdots, k\}$ to the auctioneer. Customers can misreport their valuations as well as their flexibility levels. A mechanism consists of an allocation rule $q$ and a payment rule $t$. The allocation rule $q$ is a mapping from the type profile space $\Theta \times \{1, 2, \cdots, k\}^N$ to the set of feasible allocation matrices $\mathcal{S}$. The payment rule $t$ is a mapping from $\Theta \times \{1, 2, \cdots, k\}^N$ to $\mathbb{R}^N$ with the $i$th component $t_i$ being the payment charged to customer $i$.

Consider a mechanism $(q,t)$ and suppose customers report valuations $r := (r_1,\ldots,r_N)$ and flexibility levels $c := (c_1, \ldots, c_N)$\footnote{Customers may not report their valuations and/or flexibility levels truthfully, so $r_i$ and $c_i$ may be different from $\theta_i$ and $b_i$, respectively.}. The mechanism then results in an allocation matrix $q(r,c)$ and payments $t(r,c)$. Let $a(b_i)$ be a $1\times M$ dimensional vector whose first $|\mathcal{B}_{b_i}|$ entries  are 1 and the rest are 0. In other words, the $j$ entry of $a(b_i)$ is given as
\begin{equation}\label{ai}
a_j(b_i) = \left\{
    \begin{array}{ll}
        1 & \; \; \text{if} \; \; 1 \le j \le |\mathcal{B}_{b_i}| \\
        0 & \; \; \text{otherwise}  
    \end{array}
\right..
\end{equation}

%Given that $\phi_i$'s are common knowledge, let us define $\psi_i(\theta)$ as:  \vspace{-7pt}
%\begin{equation} \label{psi}
%\psi_i(\theta) = \sum\limits_{j \in \phi_i} q_{ij}(\theta)
%\end{equation}
%Note that since $\phi_i$'s are common knowledge, they do not need to be passed to the set of $\psi_i(.)$'s arguments. 
%\begin{comment}
%From \eqref{psi} it can be seen that $\psi_i(\theta)$ is equal to 1 if user $i$ is allocated some item from her flexibility set $\phi_i$ and 0 otherwise. 
%\end{comment}
Customer $i$'s utility function can then be written in terms of its true valuation $\theta_i$, true flexibility level $b_i$, the reported valuations $r$  and the reported flexibility levels $c$ as
\begin{equation} \label{Util}
u_i(\theta_i, r, b_i, c) = \theta_i \; a(b_i) \; q_i^T(r,c) - t_i(r,c),
\end{equation}
where $q_i(r,c)$ is the $i^{\text{th}}$ row of the allocation matrix $q(r,c)$.
\vspace{-.8em}
\subsection{Incentive Compatibility and Individual Rationality} \label{sec:incentive}
The auctioneer's objective is to find a mechanism that maximizes its expected revenue while satisfying \textit{Bayesian Incentive Compatibility} and \textit{Individual Rationality} constraints. We describe these constraints below.

In a Bayesian incentive compatible (BIC) mechanism, truthful reporting of private information (valuations and flexibility levels in our setup) constitutes an equilibrium of the Bayesian game induced by the mechanism. In other words, each customer would prefer to report its true valuation and flexibility level provided that all other customers have adopted truth-telling strategy.  Bayesian incentive compatibility can be described by the following constraint:
%\begin{enumerate}
%\item BIC constraint for misreports of both valuation and flexibility level:
\begin{equation} 
\begin{split}
&\mathbb{E}_{\theta_{-i}, b_{-i}}\Big[ \theta_i \; a(b_i) \; q_i^T(\theta,b) - t_i(\theta,b) \Big] \: \ge \\
&\mathbb{E}_{\theta_{-i}, b_{-i}}\Big[ \theta_i \; a(b_i) \; q_i^T(r_i, \theta_{-i}, c_i, b_{-i}) - t_i(r_i, \theta_{-i}, c_i, b_{-i}) \Big], \\
&\forall \theta_i , r_i \in \Theta_i  \; , \; c_i , b_i \in \{1, 2, \cdots, k\} \; , \; \forall i \in \mathcal{N}. 
\end{split}
\raisetag{3.4\baselineskip}
\label{BIC}
\end{equation}
\eqref{BIC} states that  the expected utility of customer $i$ with type $(\theta_i,b_i)$  if it reports its type truthfully is greater than or equal to  its utility if it reports some other type $(r_i,c_i)$.
%Since customers' private information includes both valuation and flexibility level, there are multiple ways in which customers can misreport their types.

%From the BIC constraint in \eqref{BIC} it can be inferred that buyer $i$ can misreport both $\theta_i$ and $b_i$ or report one of them truthfully; hence the general BIC constraint in \eqref{BIC} comprises the special cases of misreporting \textit{only} valuation or flexibility level.
%\item BIC constraint for misreports of valuation:
%\begin{equation} \label{BIC_theta}
%\begin{split}
%&\mathbb{E}_{\theta_{-i}, b_{-i}}\Big[ \theta_i \; a_i(b_i) \; q_i^T(\theta,b) - t_i(\theta,b) \Big] \: \ge \\
%&\mathbb{E}_{\theta_{-i}, b_{-i}}\Big[ \theta_i \; a_i(b_i) \; q_i^T(r_i, \theta_{-i}, b) - t_i(r_i, \theta_{-i}, b) \Big] \; \\
%&\forall \theta_i , r_i \in \Theta_i \; , \; \forall b_i \in \{1, 2, \cdots, k\} \; , \; \forall i \in \mathcal{N}.
%\end{split}
%\end{equation}
%\item BIC constraint for misreports of flexibility level:
%\begin{equation} \label{BIC_b}
%\begin{split}
%&\mathbb{E}_{\theta_{-i}, b_{-i}}\Big[ \theta_i \; a_i(b_i) \; q_i^T(\theta,b) - t_i(\theta,b) \Big] \: \ge \\
%&\mathbb{E}_{\theta_{-i}, b_{-i}}\Big[ \theta_i \; a_i(b_i) \; q_i^T(\theta, c_i, b_{-i}) - t_i(\theta, c_i, b_{-i}) \Big] \; \\
%&\forall \theta_i \in \Theta_i \; , \;  c_i \le b_i \; , \; c_i , b_i\in \{1, 2, \cdots, k\} \; , \; \forall i \in \mathcal{N}.
%\end{split}
%\end{equation}
%\end{enumerate}
Individual Rationality (IR) constraint implies that each customer's  expected utility at the truthful reporting equilibrium is non-negative. This can be expressed as: \vspace{-5pt}
\begin{equation} \label{IR}
\begin{split}
&\mathbb{E}_{\theta_{-i}, b_{-i}}\Big[ \theta_i \; a(b_i) \; q_i^T(\theta,b) - t_i(\theta,b) \Big] \; \ge \; 0  \; \; , \\
&\forall \theta_i \in \Theta_i \; , b_i \in \{1, 2, \cdots, k\} \; , \; \forall i \in \mathcal{N}.
\end{split}
\end{equation}
%Now the problem is to find revenue maximizing mechanisms with feasible allocations that satisfy Bayesian incentive compatibility and voluntary participation constraints as stated in Equations \eqref{BIC} and \eqref{IR} respectively. From an optimization standpoint, the problem can be established as: \vspace{-7pt}
The expected revenue under a BIC and IR mechanism is $\mathbb{E}_{\theta, b}\Big\{\sum\limits_{i=1}^N t_i(\theta, b)\Big\}$ when all customers adopt the truthful strategy.
The auction design problem  can now be formulated as
\begin{equation*}
\begin{split}
\max\limits_{(q,t)} \; \; \; \mathbb{E}_{\theta, b}\Big\{\sum\limits_{i=1}^N t_i(\theta, b)\Big\} \;  , \; \; 
\text{subject to} \; \; \text{\eqref{BIC}, \eqref{IR}}.
\end{split}
\end{equation*}
\vspace{-2em}
\subsection{Key Assumptions} \label{se:assump}
We make two assumptions for the auction design problem. 
Firstly, we assume that the allocation rule $q$ does not give a customer any  good that is outside its \textit{reported} flexibility set.  This can be formalized as follows: 
\begin{assum}\label{assum:Nonz}
We assume that for each $i \in \mathcal{N}$, $q_i(r, c)$ can have non-zero entries \textit{only} in its first $|\mathcal{B}_{c_i}|$ positions. (Recall that the first $|\mathcal{B}_{c_i}|$ positions of $q_i(\cdot,\cdot)$ correspond  to goods in $\mathcal{B}_{c_i}$.)
\end{assum}
%(Recall that the first $|\mathcal{B}_l|$ goods belong to $\mathcal{B}_l$.)
The above assumption simply means that the mechanism respects the customers' reported flexibility constraints. We further assume that customers cannot over-report their flexibility level:
\begin{assum}\label{assum:feasb}
For each $i \in \mathcal{N}$, customer $i$'s reported flexibility level $c_i$ cannot exceed its true flexibility level $b_i$.
\end{assum}
The above assumption can be justified by noting that customers gain no utility from getting a good outside their true flexibility set and may in fact suffer a significant disutility if allocated a good outside their true flexibility set. For instance, consider the example of rate-constrained energy delivery in electricity markets that is discussed in Section I-A. While some consumers may be able to tolerate variable rates of energy delivery and are thus considered to be more flexible, other (less flexible) consumers may need to receive energy at a constant rate as their devices could be damaged otherwise. It is thus reasonable  to assume that in this case the consumers will not report higher flexibility level as it could cause \textit{significant disutility} to them.  More generally, customers may reasonably restrict themselves to under-reporting or truthfully reporting their flexibility level if goods outside their flexibility set may be damaging or cause large disutility to them.
%To avoid the risk of getting an unusable or damaging good, customers may reasonably restrict themselves to under-reporting or truthfully reporting their flexibility level.   
Assumption 2  implies that the BIC constraint in \eqref{BIC} need not consider the case of  $c_i > b_i$.

%\begin{comment}
%In an incentive compatible mechanism, truthful reporting of types constitutes an equilibrium of the Bayesian game induced among the buyers. In other words, each player would prefer to report his true type provided that all other players have adopted truth-telling strategy. 
%
%Individual Rationality or Voluntary Participation constraint requires the mechanism to benefit the users in a way that at the truthful Bayesian Nash equilibrium of the game, users will be better off participating in the mechanism rather than not taking part in it. 
%\end{comment}
\vspace{-.9em}
\subsection{Examples}
%\textcolor{blue}{TO BE EDITED}
\begin{enumerate}[1.]
\item Consider the case where $\mathcal{N}=\{1,2\}$, $\mathcal{M}=\{1,2\}$, that is, there are two customers and two goods. Let $\mathcal{B}_1=\{1\}$ and $\mathcal{B}_2=\{1,2\}$. Customer $1$'s type is  $(\theta_1=1,b_1=2)$ with probability $1$.   Customer $2$'s type $(\theta_2,b_2)$ is uniformly distributed over the set $[0.5,2] \times \{1,2\}$. Consider a mechanism for this case that operates as follows: \begin{enumerate}[(i)]
\item Each customer reports a valuation and  a flexibility level.
\item If customer $i$ has the highest valuation (assume that ties are resolved randomly), the mechanism allocates a good to customer $i$ from its reported flexibility set and charges it the second highest reported valuation.
\item The other customer is allocated a good from its flexibility set if such a good is available and it is charged a reserve price of $0.5$.
\end{enumerate}
Suppose that customer $1$ reports its type truthfully and that customer $2$'s true type is $(\theta_2=2,b_2=2)$. If customer $2$ also reports its type truthfully, it will obtain a good at a price of $1$  (the second highest reported valuation) resulting in a net utility of $2-1=1$. On the other hand, if it misreports its type as $(0.5,2)$, it will obtain a good at a price of  $0.5$ resulting in a net utility of $1.5$.

\item  Consider the same setup as above but with the following mechanism:
\begin{enumerate}[(i)]
\item Each customer reports a valuation and a flexibility level.
\item If customer $i$ has the highest valuation (assume that ties are resolved randomly), the mechanism allocates a good to customer $i$ from its reported flexibility set. Customer $i$ is charged the reported valuation of the other customer \emph{if the two reported the same flexibility level}, otherwise it pays a reserve price of $0.5$. 
\item The other customer is allocated a good from its flexibility set if such a good is available and it is charged a reserve price of $0.5$.
\end{enumerate}
Suppose that customer $1$ reports its type truthfully and that customer $2$'s true type is $(\theta_2=2,b_2=2)$. If customer $2$ also reports its type truthfully, it will obtain a good at a price of $1$ resulting in a net utility of $2-1=1$. On the other hand, if it misreports its type as $(2,1)$, it will obtain a good at a price of  $0.5$ resulting in a net utility of $1.5$.
\end{enumerate}
Thus, in both the examples above, the mechanism described  is not incentive compatible.
%We define $U(\theta_i , r_i)$ as customer $i$'s expected utility under the mechanism $(q,t)$ when it reports $r_i$ and its true valuation is $\theta_i$, that is, 
%\begin{equation} \label{Ui}
%U(\theta_i , r_i) \coloneqq \mathbb{E}_{\theta_{-i}}\Big[ \theta_i \;  \psi_i(r_i, \theta_{-i}) - t_i(r_i, \theta_{-i}) \Big]
%\end{equation}
\vspace{-1em}
\section{Characterization of BIC and IR Mechanisms}\label{sec:BIC-IR-Mechs}
Suppose all customers other than $i$ report their valuations and flexibility levels truthfully. 
We can then define customer $i$'s expected allocation and payment under the mechanism $(q,t)$ when it reports $r_i \in \Theta_i\; , \; c_i \in \{1, 2, \cdots, k\}$ as:
\begin{equation} \label{Qi}
Q_i(r_i, c_i) \coloneqq \mathbb{E}_{\theta_{-i}, b_{-i}}\Big[ q_i(r_i, \theta_{-i}, c_i, b_{-i}) \Big], 
\end{equation}
\vspace{-7pt}
\begin{equation} \label{Ti}
T_i(r_i, c_i) \coloneqq \mathbb{E}_{\theta_{-i}, b_{-i}}\Big[ t_i(r_i, \theta_{-i}, c_i, b_{-i}) \Big].
\end{equation}
%Thus, $U(\theta_i , r_i)$ is user $i$'s expected utility when her valuation for the items in $\phi_i$ is $\theta_i$ and she is reporting $r_i$ while all other buyers report their valuations truthfully. Similarly, we can define expected resource allocation and expected payment for user $i$ as: \vspace{-5pt}
%\begin{equation} \label{Qi}
%Q_i(r_i) \coloneqq \mathbb{E}_{\theta_{-i}}\Big[ \psi_i(r_i, \theta_{-i}) \Big] \; , \; \theta_i , r_i \in \Theta_i
%\end{equation}
%\vspace{-7pt}
%\begin{equation} \label{Ti}
%T_i(r_i) \coloneqq \mathbb{E}_{\theta_{-i}}\Big[ t_i(r_i, \theta_{-i}) \Big] \; , \; \theta_i , r_i \in \Theta_i
%\end{equation}
%\begin{comment}
%$Q_i(r_i)$ is the expected number of goods that user $i$ gets when she reports $r_i$ as her valuation while all other buyers report their valuations truthfully. $T_i(r_i)$ is the expected payment that user $i$ is charged when she reports $r_i$ while all other customers have adopted truth-telling strategy. 
%\end{comment}

We can now rewrite equations \eqref{BIC} and \eqref{IR} in terms of the interim quantities defined in \eqref{Qi}-\eqref{Ti}. 
%\begin{enumerate}
%\item BIC constraint for misreports of both valuation and flexibility level:
The BIC constraint for misreporting valuations and flexibility levels becomes:
\begin{equation} \label{BIC_i_QT}
\begin{split}
&\theta_i \; a(b_i) \; Q^T_i(\theta_i, b_i) - T_i(\theta_i, b_i) \ge  \\
&\theta_i \; a(b_i) \; Q^T_i(r_i, c_i) - T_i(r_i, c_i) \;  , \\
&\forall \theta_i , r_i \in \Theta_i \; , \; c_i \le b_i \; , \; c_i , b_i \in \{1, 2, \cdots, k\} \; , \; \forall i \in \mathcal{N}. 
\end{split}
\raisetag{2\baselineskip}
\end{equation}
The IR constraint is rewritten as:
\begin{equation}\label{IR_i_QT}
\begin{split}
&\theta_i \; a(b_i) \; Q^T_i(\theta_i, b_i) - T_i(\theta_i, b_i) \ge 0 \; , \\
&\forall \theta_i \in \Theta_i \; , b_i \in \{1, 2, \cdots, k\} \; , \; \forall i \in \mathcal{N}.
\end{split}
\end{equation}
%\end{enumerate}
\vspace{-2em}
\subsection{One-Dimensional Misreports} \label{sec:1D}
The  BIC constraint in \eqref{BIC_i_QT} captures all possible ways in which a customer may misreport its private information. It includes the following two special sub-classes of constraints:
\begin{enumerate}
\item BIC constraint for misreporting  only valuation:
\begin{equation} \label{BICTheta_i_QT}
\begin{split}
&\theta_i \; a(b_i) \; Q^T_i(\theta_i, b_i) - T_i(\theta_i, b_i) \ge  \\
&\theta_i \; a(b_i) \; Q^T_i(r_i, b_i) - T_i(r_i, b_i) \;  , \\
&\forall \theta_i , r_i \in \Theta_i \; , \; \forall b_i \in \{1, 2, \cdots, k\} \; , \; \forall i \in \mathcal{N}. 
\end{split}
\end{equation}
%where flexibility level is reported truthfully. 
\item BIC constraint for misreporting only  flexibility level:
\begin{align} \label{BICb_i_QT}
&\theta_i \; a(b_i) \; Q^T_i(\theta_i, b_i) - T_i(\theta_i, b_i) \ge  \notag \\
&\theta_i \; a(b_i) \; Q^T_i(\theta_i, c_i) - T_i(\theta_i, c_i) \;  , \\
&\forall \theta_i \in \Theta_i \; , \; c_i \le b_i \; , \; c_i , b_i \in \{1, 2, \cdots, k\} \; , \; \forall i \in \mathcal{N}. \notag
\end{align}
%where valuation is reported truthfully. 
\end{enumerate}

%In a similar fashion, one can characterize Individual Rationality constraint for user $i$ as follows: \vspace{-5pt}
%\begin{equation} \label{IRi}
%\begin{split}
%U_i(\theta_i, \theta_i) &\ge 0  \; \Longrightarrow \; 
%\theta_i Q_i(\theta_i) - T_i(\theta_i) \ge 0 \; , \; \theta_i \in \Theta_i
%\end{split}
%\end{equation}
%Using this characterization of customers' expected utility in terms of $Q_i(\cdot, \cdot)$ and $T_i(\cdot, \cdot)$, we can now simplify the incentive constraints that the mechanism must satisfy.
The following result relates the above constraints for ``one-dimensional'' misreports to the general BIC constraint in \eqref{BIC_i_QT}.
\begin{lemma}\label{lem:BIC2D} The BIC constraint for misreporting both valuation and flexibility level  implies and is implied by the BIC constraints for misreporting only valuation and misreporting only flexibility level. That is, \eqref{BIC_i_QT} holds if and only if \eqref{BICTheta_i_QT} and \eqref{BICb_i_QT} hold.
\begin{proof} See Appendix \ref{sec:BICLemPrf}. \end{proof}
\end{lemma}
Lemma \ref{lem:BIC2D} allows us to replace the general  BIC constraint for two-dimensional misreports by the simpler one-dimensional BIC constraints given in \eqref{BICTheta_i_QT} and \eqref{BICb_i_QT}.   The auction design problem now becomes:
\begin{equation*} 
\begin{split}
\max\limits_{(q ,t)} \; \; \mathbb{E}_{\theta, b}\Big\{\sum\limits_{i=1}^N t_i(\theta, b)\Big\} \; , \; 
\text{subject to} \; \text{\eqref{IR_i_QT}}, \eqref{BICTheta_i_QT}, \eqref{BICb_i_QT}.
\end{split}
\end{equation*}
\vspace{-2em}
\subsection{Alternative Characterization of \eqref{IR_i_QT}, \eqref{BICTheta_i_QT}, \eqref{BICb_i_QT}} \label{sec:alter}
We will now derive alternative characterizations of the constraints \eqref{IR_i_QT}, \eqref{BICTheta_i_QT}, \eqref{BICb_i_QT} that will be helpful for finding the optimal mechanism.
 
%We can now derive  a necessary and sufficient condition for a direct-revelation mechanism to be Bayesian incentive compatible and individually rational \textit{in valuation}. 
\begin{lemma}\label{thm:one}
 A mechanism $(q,t)$ satisfies the BIC constraint for misreporting only valuation (as given in \eqref{BICTheta_i_QT}) if and only if for all $i \in \mathcal{N}$, $a(b_i) Q^T_i(r_i, b_i)$ is non-decreasing in $r_i$ for all $b_i$ and  \vspace{-5pt}
\begin{align}
&T_i(r_i, b_i)  \notag \\
&= K_i(b_i) + r_i a(b_i) Q^T_i(r_i, b_i) - a(b_i) \int\limits_{\theta_i^{\text{min}}}^{r_i} Q^T_i(s, b_i) \: ds, \label{Thrm1}
\end{align}
for all $r_i,b_i$.
\end{lemma}
\begin{proof}
%The proof is similar to the arguments in chapters 2-3 of \cite{borgers2015introduction} for characterizing BIC and IR mechanisms.
See Appendix \ref{sec:thmBICValPrf}.
\end{proof}

\begin{lemma}\label{lemma:ir}
Suppose the mechanism $(q,t)$ satisfies the BIC constraint for misreporting only valuation (as given in \eqref{BICTheta_i_QT}). Then, it satisfies the IR constraint \eqref{IR_i_QT} if and only if for all $b_i$
\begin{equation} \label{Ki}
\theta_i^{\text{min}} \; a(b_i) \; Q^T_i(\theta_i^{\text{min}}, b_i) - T_i(\theta_i^{\text{min}}, b_i)   \ge 0.
\end{equation}
\end{lemma}
\begin{proof}
Clearly \eqref{IR_i_QT} implies \eqref{Ki}. The converse follows from Lemma \ref{thm:one} by noting that \[K_i(b_i) = T_i(\theta_i^{min}, b_i) - \theta_i^{min} a(b_i) Q^T_i(\theta_i^{min}, b_i),\]
and that the right hand side above is non-positive due to \eqref{Ki}.
\end{proof}

Using the  above two lemmas, we  derive a sufficient condition for the mechanism to satisfy the BIC constraint for misreporting only flexibility level.
\begin{lemma}\label{lem:BICSuffb} Suppose the mechanism $(q,t)$ is individually rational and satisfies the BIC constraint for misreporting only valuation (as given in \eqref{BICTheta_i_QT}). Then the mechanism $(q,t)$ satisfies the BIC constraint for misreporting only flexibility level if the following are true:
\begin{enumerate}[(i)]
\item $a(c_i)\;Q^T_i(\theta_i,c_i)$ is non-decreasing in $c_i \; , \; \forall \theta_i \in \Theta_i, \forall i \in \mathcal{N}$, and
\item $T_i(\theta_i^{\text{min}}, c_i) = 0 \; , \; \forall c_i \in \{ 1, 2, \cdots, k \}\; , \; \forall i \in \mathcal{N}$. 
\end{enumerate}
\begin{proof} See Appendix \ref{sec:SuffbPrf}. \end{proof}
\end{lemma}
%%%%% INPUT REVENUE MAXIMIZATION SECTION %%%%%%%%%
%\input{revMax_brief}
\vspace{-2em}
\section{Revenue Maximizing Mechanism}\label{sec:RMM}
We can now use the results of Section \ref{sec:BIC-IR-Mechs} to simplify the objective of the auction design problem. 
We define  
\begin{equation}
w_i(\theta_i, b_i) := \Big( \theta_i - \frac{1 - F_{i}(\theta_i | b_i)}{f_{i}(\theta_i | b_i)}\Big),
\end{equation}
 where  $f_{i}(\theta_i | b_i)$ is the conditional probability density function of customer $i$'s valuation conditioned on its flexibility level $b_i$ and $F_{i}(\theta_i | b_i)$ is the corresponding cumulative distribution function. $w_i(\theta_i, b_i)$  is referred to as customer $i$'s \textit{virtual type or virtual valuation} in economics terminology \cite{borgers2015introduction}.

\begin{lemma} \label{lemma:revenue}
%Given the results of Lemmas \ref{thm:one}-\ref{lem:BICSuffb}, the  problem
%\begin{equation}\label{prob1}
%\max\limits_{(q,t)} \;\; \mathbb{E}_{\theta, b}\Big\{\sum\limits_{i=1}^N t_i(\theta, b)\Big\},
%\end{equation}
%is the same as the following problem
 Suppose $(q,t)$ is a BIC and IR mechanism for which (i) $K_i(b_i)=0$ for all $i$ and $b_i$\footnote{ $K_i(b_i)$ appears in Lemma \ref{thm:one}.} and (ii) $q$ is a solution to the following functional optimization problem
\begin{equation}\label{prob2}
\max\limits_{q} \;\; \sum\limits_{b} \int_{\theta} \; \sum\limits_{i=1}^N \Big[ a(b_i) \; q_i^T(\theta, b) w_i(\theta_i, b_i)\Big] f(\theta, b) d\theta.
\end{equation}
Then $(q,t)$ is an optimal mechanism.
\begin{proof}
See Appendix \ref{sec:revLemmaPf}.
\end{proof}
\end{lemma}

 In order to simplify the maximization problem in \eqref{prob2}  we assume that
  the virtual types $\Big( \theta_i - \frac{1 - F_{i}(\theta_i | b_i)}{f_{i}(\theta_i | b_i)}\Big)$ are non-decreasing in $\theta_i$ and $b_i$. Such a condition holds if $\frac{f_{i}(\theta_i | b_i)}{1 - F_{i}(\theta_i | b_i)}$ is non-decreasing in $\theta_i$ and $b_i$. 
  
%  This condition  can be viewed as a generalization of the increasing hazard rate condition \cite[Chapter 2]{borgers2015introduction} and is similar to the condition about monotonicity of virtual valuations  described in \cite{pai2013optimal} for multidimensional private types. We formally state this condition  and our assumptions below.

\textit{Generalized Monotone Hazard Rate Condition:} The type $(\theta_i, b_i)$ is said to be partially ordered above $(\theta_i', b_i')$, and this relation denoted by  $(\theta_i, b_i) \succeq (\theta_i', b_i')$, if $\theta_i \ge \theta_i'$ and $b_i \ge b_i'$. The distribution $f_{i}(\cdot, \cdot)$  satisfies the generalized monotone hazard rate condition if: 
\begin{equation}\label{MHRC}
\begin{split}
 \hspace*{-0.3cm}(\theta_i, b_i) \succeq (\theta_i', b_i') \; \; \Longrightarrow \; \; \frac{f_{i}(\theta_i | b_i)}{1 - F_{i}(\theta_i | b_i)} \ge \frac{f_{i}(\theta_i' | b_i')}{1 - F_{i}(\theta_i' | b_i')};
\end{split}
%\raisetag{\baselineskip}
\end{equation}
Further,  $b_i > b'_i$ and $\theta_i \geq \theta_i'$ imply
\begin{equation}\label{MHRC_2} \frac{f_{i}(\theta_i | b_i)}{1 - F_{i}(\theta_i | b_i)} > \frac{f_{i}(\theta_i' | b_i')}{1 - F_{i}(\theta_i' | b_i')}.
\end{equation}
   
\begin{assum}\label{assum:MHRC}
We assume that the probability density functions $f_i(\cdot, \cdot)$ satisfy the generalized monotone hazard rate condition for all $i \in \mathcal{N}$.
\end{assum}

\begin{remark}\label{rem:Hz}
The above condition  can be viewed as a generalization of the increasing hazard rate condition \cite[Chapter 2]{borgers2015introduction} and is similar to the condition about monotonicity of virtual valuations  described in \cite{pai2013optimal} for multidimensional private types. To get further insights into the implications of Assumption \ref{assum:MHRC}, consider two possible types of customer $i$ with flexibility levels $b_C \ge b_D$. Let $X_C$ and $X_D$ be  random variables that are distributed according to the corresponding conditional probability density functions $f_{i}(.|b_C)$ and $f_{i}(.|b_D)$ respectively. For a realization $\alpha$ of $X_C, X_D$,  Assumption \ref{assum:MHRC} implies:
\begin{equation}
\frac{f_{i}(\alpha|b_C)}{1 - F_{i}(\alpha|b_C)} \ge \frac{f_{i}(\alpha|b_D)}{1 - F_{i}(\alpha|b_D)} \; .
\end{equation}  
In the language of Shaked and Shanthikumar \cite[Section 1.B.1]{shaked2007stochastic}, this means that $X_C$ is smaller than $X_D$ in the hazard rate order (denoted by $X_C \le_{hr} X_D$) . According to Theorem 1.B.1. in  \cite{shaked2007stochastic}, it can then be concluded that $X_D$  stochastically dominates $X_C$ in the first order (denoted by $X_C \le_{st} X_D $). Thus one can roughly say that less flexible consumers are expected to have higher valuations than more flexible consumers. 
\end{remark}

%Here we assume that customers with minimum valuation have negative virtual valuations; this is formalized as follows:

\begin{assum} \label{assum:Negwmin}
We assume that $w_i(\theta_i^{min}, b_i) < 0 \; , \; \forall b_i \in \{1, 2, \cdots, k\} \; , \; \forall i \in \mathcal{N}$.
\end{assum}

\begin{remark}\label{rem:wMinfamilies}
Two examples of   families of probability distributions $f_{i}(\theta_i | b_i)$ that satisfy Assumption \ref{assum:Negwmin} are:
\begin{itemize}
\item Let $\theta_i^{min} = 0$. Then any probability density function $g(\cdot)$ on $[\theta_i^{min}, \theta_i^{max}]$ such that $g(0) > 0$ satisfies Assumption \ref{assum:Negwmin}.
\item A uniform distribution $g(\theta_i) = \frac{1}{|\theta_i^{max} - \theta_i^{min}|}, \theta_i \in [\theta_i^{min}, \theta_i^{max}]$ with $\theta_i^{min} < |\theta_i^{max} - \theta_i^{min}|$ satisfies Assumption \ref{assum:Negwmin}.
\end{itemize}
\end{remark}
The following theorem characterizes the optimal mechanism under the above assumptions.
\begin{theorem}\label{thm:2}
 Consider the allocation and tax functions $(q^*,t^*)$ defined below \vspace{-0.3cm}
\begin{align} \label{qi_opt}
q^*(\theta, b) \in \argmax\limits_{\mathbf{A} \in \mathcal{S}} \; &\sum\limits_{i=1}^N \left(a(b_i)  \mathbf{A}_i^T\right) w_i(\theta_i, b_i) \; ,
\end{align} 
where $\mathbf{A}_i$ is the $i$th row of matrix $\mathbf{A}$;
\begin{equation} \label{ti_opt}
t^*_i(\theta, b) := \theta_i \: a(b_i) \: q_i^{*T}(\theta, b) - a(b_i) \: \int\limits^{\theta_i}_{\theta_i^{min}} q_i^{*T}(s,\theta_{-i}, b) \: ds.
\end{equation}
%such that for all $b_i \in \{1, 2, \cdots, k\}$ and $\forall \theta_{-i} \in \Theta_{-i}$, 
%\begin{equation} \label{qmin}
%\begin{split}
%a(b_i) \: q_i^{*T}(\theta_i^{min}, \theta_{-i}, b) = 0. 
%\end{split}
%\end{equation}
Then, under Assumptions 1-4, $(q^*,t^*)$ is a revenue-maximizing  Bayesian incentive compatible and individually rational mechanism.

%,  i.e., $(q^*,t^*)$ is a solution of the maximization problem in \eqref{Opt}. 

\end{theorem}
\begin{proof} See Appendix \ref{sec:Thm2Pf}. \end{proof}
The optimal allocation matrix $q^*(\theta, b)$ given in \eqref{qi_opt} is the solution of an integer program and hence computationally hard to obtain. Moreover, each type profile $(\theta, b) \in \Theta \times \{1, 2, \cdots, k\}^N$ requires the solution of a different integer program. Similarly, the characterization of payments given by \eqref{ti_opt} is not very useful from a computational viewpoint as it requires the solution of a continuum of integer programs. In the next section, we leverage the nested structure imposed on  customers' flexibility sets  to simplify the optimal mechanism. 
%~\\
%Note that in order to maximize the weighted sum in this integer program (Equations \eqref{qi_opt}-\eqref{ti_opt}) we remove the terms with negative weights ($w_i < 0$), i.e., users with negative virtual types will not be allocated anything. As mentioned earlier, virtual type expression (Equation~\eqref{wi}) is increasing in $\theta_i$ and hence, there exists some threshold valuation for each user beyond which their virtual type is ensured to be positive. This threshold valuation is referred to as \textit{reserve price}; let us denote user $i$'s reserve price as $\theta_i^{\text{res}}$. Therefore user $i$ needs to have $\theta_i > \theta_i^{\text{res}} \; \Rightarrow \; w_i > 0$ in order to ensure that she will be \textit{considered} for allocation in the first place. Users with $w_i < 0$ will not be qualified to receive service. 
%
%The solution to the optimization problem formulated in Theorem 2 (Equations~\eqref{qi_opt}-\eqref{ti_opt}), constitutes an incentive compatible and individually rational mechanism that maximizes auctioneer's total expected revenue. In sequel, we show how a certain structure of the flexibility sets can be leveraged to simplify the solution to the integer program stated in Equations~\eqref{qi_opt}-\eqref{ti_opt}.
\vspace{-1.8em}
\section{A Candidate Revenue Maximizing Mechanism} \label{sec:NestedPhis} 
Based on their true flexibility sets, we can divide the customers into $k$ classes: $\mathcal{C}_l$ is the set of customers with flexibility set $\mathcal{B}_l$. Clearly, $\mathcal{N} = \bigcup\limits_{i=1}^k \mathcal{C}_i$ and for $i \neq j$, $\mathcal{C}_i \cap \mathcal{C}_j = \emptyset$. We define
\begin{equation}  \label{mi}
\begin{split}
 n_l &:= |\mathcal{C}_l|, ~l =1,\ldots, k,\\
m_l &:= |\mathcal{B}_l \: \backslash \: \mathcal{B}_{l-1}| ,  \; l =2,\ldots, k   , \; \; 
m_1 := |\mathcal{B}_1|.
\end{split}
\end{equation}

We also define the vectors \textbf{n} and \textbf{m} as
\begin{equation} \label{SupplyDemandProfs}
\begin{split}
\textbf{n} := (n_1, n_2, \cdots, n_k) \;  , \; \; 
\textbf{m} := (m_1, m_2, \cdots, m_k).
\end{split}
\end{equation}
The vector \textbf{n} is referred to as the \emph{demand profile} and the vector \textbf{m} is referred to as the \emph{supply profile}.
\vspace{-1em}
\subsection{Supply Adequacy Problem} \label{sec:BinServProb}

Before describing the optimal mechanism, we will need to answer two questions:
\begin{enumerate}
\item Given a supply profile \textbf{m} and a demand profile \textbf{n}, can the available goods be used to satisfy all customers? In other words, does there exist an allocation matrix $\textbf{A}\in\{0,1\}^{N\times M}$ such that \vspace{-.8em}
\begin{equation} \label{ServeAll}
\begin{split}
\sum\limits_{j \in \phi_i} \: \textbf{A}(i,j) = 1, ~~\forall \; i \in \mathcal{N}, \; \; 
\sum\limits_{i=1}^{N} \: \textbf{A}(i,j) \le 1 , ~~ \forall \; j \in \mathcal{M}.   
\end{split}
\end{equation}
The above conditions on \textbf{A} ensure that each customer gets a good from its flexibility set and that a good is not allocated to multiple customers. If such an allocation matrix exists, we say that the supply profile \textbf{m} is adequate for the demand profile \textbf{n}.

\item If the supply profile \textbf{m} is not adequate for the demand profile \textbf{n}, we  have to remove some customers from the demand profile to achieve adequacy. What is the minimum number of customers that must be removed to achieve adequacy?
\end{enumerate}

Borrowing ideas from \cite{7098381}, we provide answers to the above questions in Lemmas \ref{lem:Adeq} and \ref{lem:MinRemov} below.

\begin{lemma} \label{lem:Adeq}
 We say that $\textbf n \prec_w \textbf m$ if the following $k$ inequalities hold: \vspace{-1em}
\begin{equation} \label{AdeqIneqs}
\begin{split}
\sum\limits_{i=1}^l n_i \; \le \; \sum\limits_{i=1}^l m_i,
~~~  l = 1, 2, \cdots, k.
\end{split}
\end{equation}
(a) The  supply profile \textbf{m} is adequate for the demand profile \textbf{n} if and only if $\textbf n \prec_w \textbf m$. 
(b) If the supply profile is adequate for the demand profile, a feasible allocation is obtained as follows: Arrange customers in order of increasing flexibility level; then the $i$th customer in this order gets the $i$th good. 
  \end{lemma}
\begin{proof} 
The proof is straightforward and therefore omitted due to space limitations. For part (b), recall that the goods are indexed such that the first $|\mathcal{B}_l|$ goods belong to $\mathcal{B}_l$, for $l=1,\ldots,k$. 
%See Appendix \ref{sec:LemAdeq}.
 \end{proof}
%\begin{comment}
%Based on the results from Lemma 1, if on any flexibility level demand exceeds the corresponding available supply, it will not be possible to serve all customers that are interested in the limited supply and hence, some of the consumers of those classes will inevitably remain unserved. In order to maximize his expected revenue, the auctioneer naturally needs to serve as many consumers as possible given the available supply. 
%\end{comment}
%Let \textbf{n} and \textbf{m} denote the demand and supply profiles respectively which are defined in the following vector forms: \vspace{-8pt}
%\begin{equation} \label{SupplyDemandProfs}
%\begin{split}
%\textbf{n} = (n_1, n_2, \cdots, n_k) \; \; , \; \; 
%\textbf{m} = (m_1, m_2, \cdots, m_k)
%\end{split}
%\end{equation}
If the supply profile \textbf{m} is not adequate, we  have to remove some customers from the demand profile. Consider a demand profile $\tilde{\textbf{n}} \leq \textbf{n}$ obtained by removing some customers. This new demand profile will result in adequacy if and only if  $\tilde{\textbf{n}} \prec_w \textbf{m}$. 
Thus, the minimum number of customers to be removed to achieve adequacy is given by the following optimization problem: \vspace{-.3cm}
%\begin{equation} \label{OptimRemovUsers}
%\begin{split}
%\min\limits_{r_1,r_2\ldots r_k} \; \; &\sum\limits_{i=1}^k r_i \\
%\text{subject to} \; \; & r_i \in \mathbb{Z}^+, ~ \forall ~i \in \mathcal{N}\\
% & \tilde{\textbf{n}} \prec_w \textbf{s} \\
%&  \tilde{\textbf{n}} =  (n_1-r_1, \ldots, n_k=r_k). 
%\end{split}
%\end{equation}
\begin{equation} \label{OptimRemovUsers}
\begin{split}
\min\limits_{\tilde{\textbf{n}}} \; \; \sum\limits_{i=1}^k (n_i-\tilde{n}_i) \; \; , \; \; 
\text{subject to} \; \;  \tilde{\textbf{n}} \prec_w \textbf{m} \; ,\;  \tilde{\textbf{n}} \; \le \textbf{n}.
\end{split}
\end{equation}
$\tilde{\textbf{n}}$ in the above optimization problem is a vector of non-negative integers. 
The above integer program has a simple solution described in the following lemma.

\begin{lemma}\label{lem:MinRemov}
 Define $r_1^* := (n_1 - m_1)^+$. For $ 2 \leq j \leq k$, recursively define $r^*_j$ as the solution of the following one-dimensional integer program: 
\begin{equation} \label{r_jOpt}
\begin{split}
r_j^* := &\argmin\limits_{r_j \in \mathbb{Z}^+} \; \; r_j \\
\text{subject to} \; \; \; \;  
&\sum\limits_{l=1}^{j-1} (n_l - r_l^*) + (n_j - r_j)\;\le\;\sum\limits_{l=1}^{j} m_l~~.
\end{split}
\end{equation}
Equivalently, $r_j^*  :=  \Big( \sum\limits_{l=1}^j (n_l - m_l) - \sum\limits_{l=1}^{j-1} r_l^* \Big)^+$. 
Then, (i) for $j=1,\ldots,k$, at least $\sum_{i=1}^j r^*_i$  customers must be removed from the first $j$ classes to satisfy the   inequalities  \eqref{AdeqIneqs} of Lemma \ref{lem:Adeq}; (ii)
%\begin{enumerate}[(i)]
$\sum_{i=1}^k r^*_i$ is the minimum value of  the integer program in \eqref{OptimRemovUsers}. 

%\item $\sum_{i=1}^k r^*_i = \max\limits_{1\le i \le k} \; \Big( \sum\limits_{j=1}^i (n_j - m_j)\Big)^+$
%\end{enumerate}

%\begin{enumerate}[(i)]
%\item The resultant profile $(r_1^*, r_2^*, \cdots, r_k^*)$ is \textit{feasible}, i.e., it does not violate the supply adequacy criteria (inequalities in \eqref{OpDef})
%\item The resultant profile $(r_1^*, r_2^*, \cdots, r_k^*)$ is optimal, i.e. $\sum\limits_{i=1}^k r_i^*$ is the \textit{minimum} number of customers that need to be removed so that supply adequacy criteria are met and: \hspace{2pt}
%$\sum\limits_{i=1}^k r_i^* \; = \; \max\limits_{1\le i \le k} \; \; \Big(\sum\limits_{j=1}^i (n_j - m_j) \Big)^+$
%\end{enumerate} 
\end{lemma}
\begin{proof} See Appendix \ref{sec:MinRemovPf}. \end{proof}
%Once excess demand is removed through the described procedure, all of the remaining consumers can be served. In the following section, we use the results of Theorem 3 to construct the optimal allocation rule and simplify users' tax functional.
\vspace{-1.5em}
\subsection{Optimal Allocation } \label{sec:AllocRule}
We can now use the results of Section \ref{sec:BinServProb} to find the optimal allocation for a given type profile $(\theta,b)$. Recall from Theorem \ref{thm:2} that the optimal allocation is given as \vspace{-.5em}
\[ q^*(\theta, b) \in \argmax\limits_{\mathbf{A} \in \mathcal{S}} \sum\limits_{i=1}^N \left(a(b_i)  \mathbf{A}_i^T\right) w_i(\theta_i, b_i) .  \]
%\textbf{Optimal Allocation:} 
We describe the optimal allocation in the following steps:
\begin{enumerate}
\item Firstly, any customer $l$ with $w_l(\theta_l, b_l) \leq 0$ is immediately removed from consideration (that is, it is not allocated any good).  Since virtual valuation is a non-decreasing function of true valuation, $w_l(\theta_l, b_l) \leq 0 $ if and only if $\theta_l \leq \theta^{res}_{l,b_l}$, where $\theta^{res}_{l,b_l}$ is a threshold based on the probability distribution of $\theta_l$ conditioned on the flexibility level $b_l$. This threshold is called the reserve price for  customer $l$ with flexibility level $b_l$. 
%Moreover, the constraint in \eqref{qmin} imposed in the problem formulated in Theorem \ref{thm:2} requires that consumers with $\theta_i = \theta_i^{min}$ have 0 allocations. Thus buyers with $\theta_i = \theta_i^{min}$ will also be removed from consideration.

For each class of customers, we define the subset of customers who have positive virtual valuations:
\begin{equation}
\mathcal{C}_i^+ := \{l \in \mathcal{C}_i : w_l(\theta_l, i) >0  \}.
\end{equation}
Let $n^+_i = |\mathcal{C}_i^+|$. Define $r^*_1,\ldots,r^*_k$ as in Lemma \ref{lem:MinRemov} by replacing $n_i$ with $n^+_i$ for all $i$.  
\item  Let $\mathcal{L}_1 := \mathcal{C}_1^+$. From $\mathcal{L}_1$, $r^*_1$ customers with the lowest virtual valuations are removed from consideration\footnote{Ties are resolved randomly. For continuous valuations, ties happen with zero probability and therefore the allocation rule for ties does not affect expected revenue.}. The set of remaining customers in $\mathcal{L}_1$ is denoted by $\mathcal{N}_1$.
\item We now proceed iteratively:  For $2 \leq i \leq k$, given the set $\mathcal{N}_{i-1}$, define  $\mathcal{L}_i := \mathcal{N}_{i-1} \bigcup \mathcal{C}^+_i$. Remove $r^*_i$ customers with lowest virtual valuations from $\mathcal{L}_i$. The set of remaining customers in $\mathcal{L}_i$ is now defined as $\mathcal{N}_i$.
\item After the $k^{th}$ iteration, all customers in $\mathcal{N}_k$ are allocated a good from their respective flexibility sets.
\end{enumerate}
%\textcolor{blue}{Discuss why is this process optimal?}

The iterative procedure described above is  outlined in  Algorithm \ref{alg:alloc}.

The optimality of the above allocation can be intuitively explained as follows: Firstly, it is clear that an optimal allocation should not give any goods to customers with non-positive virtual valuations. Among the remaining customers of class $\mathcal{C}_1$, at least $r^*_1$ customers cannot be served (see Lemma \ref{lem:MinRemov} with $n_i$ replaced by $n^+_i$ for all $i$). It is easy to see that the $r^*_1$ customers with the lowest virtual valuations should be removed. This argument can be used iteratively. At the $i$th iteration, at least $r^*_i$ additional customers need to be removed from the first $i$ classes otherwise the $i$th adequacy inequality would be violated. An optimal allocation should remove $r^*_i$ customers with lowest virtual valuations. After the $k$th iteration, exactly $\sum_{i=1}^k r^*_i$ customers have been removed and the remaining customers' demand profile satisfies all the adequacy inequalities.

%In order to maximize seller's revenue, consumers with higher valuations will be given higher chances of receiving a good. Given the numbers $r_i^* \; , \; i = 1, 2, \cdots, k$ characterized as the solution to the optimization problem in \eqref{OptimCumulRemov}, at iteration $i$ of the described procedure, the seller will thus remove $r_i^*$ consumers with the \textit{lowest} valuations among all the users in $\bigcup\limits_{j=1}^{i-1}\tilde{\mathcal{C}_j} \cup \mathcal{C}_i$ where $\tilde{\mathcal{C}_j}$ denotes the modified class $\mathcal{C}_j$ after possible shrinkages of its size during the first $i-1$ iterations. 
%
%Let us define the set $\mathcal{N}_i = \bigcup\limits_{j=1}^i \tilde{\mathcal{C}_j}$ and $\mathcal{L}_i = \mathcal{N}_{i-1} \bigcup \mathcal{C}_i$. 

The above optimal allocation procedure can also be described using   thresholds. Define
%\begin{equation}
%\theta_1^{thr} := (r_1^*)^{\text{th}} \; \text{lowest valuation in} \; \; \mathcal{C}_1^+ ,
%\end{equation}
\vspace{-.4em}
\begin{equation} \label{ThetaThr}
w_i^{thr} := (r_i^*)^{\text{th}} \; \text{lowest virtual valuation in} \; \mathcal{L}_i , i =  1,2, \cdots, k.
\end{equation}
(If $r^*_i=0$, $w_i^{thr}=0$.)\\
Then, at iteration $i$, customers that have virtual valuations less than or equal to $w_i^{thr}$  will be removed from the set $\mathcal{L}_i$. 
\begin{algorithm} [ht]
  \caption{Pseudocode for Computing the Optimal Allocation}
  \label{alg:alloc}
  \begin{algorithmic}[1]
\State   Remove all  consumers with $\; w_l(\theta_l, b_l) \le 0$.
\State Define $\; \mathcal{C}_i^+ \coloneqq \{ l \in \mathcal{C}_i : w_l(\theta_l, i) > 0 \},$ and let $n_i^+ = |\mathcal{C}_i^+|$ for $i = 1, 2, \cdots, k$. 
\State Compute $r_1^*, \cdots, r_k^*$ through: 

$\; \; r_j^*  =  \Big( \sum\limits_{l=1}^j (n^+_l - m_l) - \sum\limits_{l=1}^{j-1} r_l^* \Big)^+ \; , \; j = 1, 2, \cdots, k$. 
\State Define $\mathcal{N}_0 := \emptyset$.
\ForAll{ $i = 1, 2, \cdots, k$}: 
\State Given the set $\mathcal{N}_{i-1}$, define $\mathcal{L}_i \coloneqq \mathcal{N}_{i-1} \cup \mathcal{C}_i^+$
\State Define $w_i^{thr} \coloneqq (r_i^*)^{\text{th}} \; \text{lowest virtual valuation in} \; \mathcal{L}_i$   
\State  Define $\mathcal{N}_i \coloneqq \{ l \in \mathcal{L}_i : w_l(\theta_l, b_l) > w_i^{thr} \}$ 
\State Keep the consumers in $\mathcal{N}_i$, remove the ones in $\mathcal{L}_i \setminus \mathcal{N}_i$  
\State $i \longleftarrow i+1$ 
\EndFor
\State All customers in $\mathcal{N}_k$ are allocated goods in order of increasing flexibility level (as per Lemma \ref{lem:Adeq}).
  \end{algorithmic}
\end{algorithm} 
%\footnotetext{ $\mathcal{N}_0 = \emptyset$.}
%
%Equivalently, removal criterion can be specified in terms of virtual valuations:
%\vspace{-5pt}
%\begin{equation} \label{wThr}
%\begin{split}
%w_i^{thr} = (r_i^*)^{\text{th}} \; &\text{lowest virtual valuation in} \; \mathcal{L}_i \; , \; 
%i = 1, 2, \cdots, k
%\end{split}
%\end{equation}
%
%Note that $w_i(\theta_i)$ is non-decreasing in $\theta_i$; hence the removal criteria in Equations \eqref{ThetaThr} and \eqref{wThr} are equivalent. 
%
%\begin{comment}
%The revenue-maximizing allocation rule as formulated in Equation \eqref{qi_opt} from Theorem 2, is the one which ensures that consumers with highest virtual valuations are served so that the linear combination of positive virtual valuations in \eqref{qi_opt}, i.e. the seller's revenue from the type profile $\theta$, is maximized. 
%\end{comment}
Under the optimal allocation,   customer $l$ in class $\mathcal{C}_i$ gets a desired good if its virtual valuation exceeds $0$ and thresholds $w_i^{thr}, w_{i+1}^{thr}, \cdots, w_k^{thr}$. Let us define
\vspace{-.5em}
\begin{equation}\label{valThrsh}
\theta_{l,i}^{thr} = \Big\{ x  : w_l(x,i) = \max\{0,w_i^{thr}, w_{i+1}^{thr}, \cdots, w_k^{thr}\} \Big\}.
\end{equation}
%Recall that $a(b_l) q^{*T}_l(\theta, b)$ is $1$ if customer $l$ gets a good from its flexibility set and $0$ otherwise. 
Because of the monotonicity of virtual valuation as a function of true valuation, customer $l$ in class $\mathcal{C}_i$ gets a good if $\theta_l > \theta_{l,i}^{thr}$. 
Thus, 
\vspace{-.7em}
\begin{equation} \label{Alloc}
\begin{split}
a(i) q_l^{*T}(\theta, b) = \left\{
    \begin{array}{ll}
        1 & \text{if} \minitab \theta_l > \theta_{l,i}^{thr}  \\
        0 & \text{otherwise}  
    \end{array}
\right., \\
\forall l \in \mathcal{C}_i, \; \;   \forall i = 1, 2, \cdots, k .
\end{split}
\end{equation}
\subsection{Payment Functions} \label{sec:TaxFuncs}
We can now use the optimal allocation rule described in section \ref{sec:AllocRule} to simplify customers' payment functions. 
%For customer $l$ in class $\mathcal{C}_i$ define $d_{il} := \max\{ \theta_{l,i}^{thr}, \theta^{res}_{l,i}\}$. 
From \eqref{ti_opt} the optimal payment function for customer $l$ in flexibility class $\mathcal{C}_i$ has the following form:
\begin{align} \label{Tax_i}
t^*_l(\theta, b) &= \theta_l a(i) q_l^{*T}(\theta, b) - a(i)  \int\limits^{\theta_l}_{\theta_l^{min}} q_l^{*T}(s, \theta_{-l}, b) \: ds. 
\raisetag{0.3\baselineskip}
\end{align}
Using the definition of  $a(i) q_l^{*T}(\theta, b)$ given in \eqref{Alloc}, $t_l(\theta, b)$ can be simplified as: 
\begin{enumerate}
\item If $\theta_l > \theta_{l,i}^{thr}$, \vspace{-1em}
\begin{equation} \label{taxFunc}
\begin{split}
t^*_l(\theta, b) &= \theta_l  - \int\limits_{\theta_l^{min}}^{\theta_{l,i}^{thr}} \underbrace{a(i)q_l^{*T}(s, \theta_{-l}, b)}_{=0} \; ds \\
&- \int\limits_{\theta_{l,i}^{thr}}^{\theta_l} \underbrace{a(i)q_l^{*T}(s, \theta_{-l}, b)}_{=1} \; ds = \theta_{l,i}^{thr}. 
\end{split}
\end{equation}
\item If $\theta_l \leq  \theta_{l,i}^{thr}$, \vspace{-.6em}
 \begin{equation}\label{pay0}
 t^*_l(\theta,b) =0.
 \end{equation}
%$t^*_l(\theta,b) =0$.
\end{enumerate}

The optimal allocation  and payments can thus be  computed through the straightforward threshold-based procedure constructed in Sections \ref{sec:AllocRule} and \ref{sec:TaxFuncs}. By using the nested structure of the flexibility sets, this procedure obviates the need to solve the computationally hard integer program formulated in Theorem \ref{thm:2}. 
\vspace{-.5em}
\begin{remark}\label{remark:expostIR}
It should be noted that  equations \eqref{Alloc}-\eqref{pay0} imply that the mechanism $(q^*,t^*)$ proposed in sections \ref{sec:AllocRule} and \ref{sec:TaxFuncs} is  \textit{ex post} individually rational, that is, at the truthful equilibrium the mechanism guarantees that each consumer gets non-negative  utility for every realization of consumers' types. In particular, under $(q^*,t^*)$ a consumer who receives no good does not pay anything. 
\end{remark}
\vspace{-1em}
\subsection{Computational Complexity of the Algorithm}\label{sec:comp-alg} 
To get a better idea of the computational complexity of the  solution approach  developed in Sections \ref{sec:AllocRule} and \ref{sec:TaxFuncs}, we take a closer look at the key steps.

The  allocation procedure  requires the recursive evaluation of the quantities $r_1^*, \cdots, r_k^*$ as per the closed-form solutions given in step 3 of Algorithm \ref{alg:alloc}.
%\begin{align}\label{eq:rstar}
%r_1^* &=  \Big( n^+_1 - m_1 \Big)^+ \\
%r_2^* &=  \Big( (n^+_1 - r_1^*) + n^+_2 - (m_1 + m_2)  \Big)^+ \\
%r_j^*  &=  \Big( \sum\limits_{l=1}^{j-1} (n^+_l - r_l^*) + n_j^+ - \sum\limits_{l=1}^{j} m_l \Big)^+ \; , \; j = 3, 4, \cdots, k.
%\end{align}
As is evident from these equations, computation of $r_1^*, \cdots, r_k^*$ is straightforward and needs nothing more than addition, subtraction and comparison with 0.

Once these quantities are obtained,  the allocation rule is  given by an iterative procedure that consists of  $k$ iterations. At each iteration the following steps are taken:
\begin{enumerate}
\item Sort the virtual valuations in the set $\mathcal{L}_i$. (\textit{sorting})
\item Find the  $r_i^*$th lowest virtual valuation in $\mathcal{L}_i$ and set $w_i^{thr}$ equal to that value. (\textit{counting}).
\item Remove from $\mathcal{L}_i$ consumers with virtual valuations less than or equal to $w_i^{thr}$ and keep the remaining ones in the new set $\mathcal{N}_i$. (\textit{deletion})
\item Move to iteration $i+1$.
\end{enumerate} 
As can be seen from the above steps, the only real computation involved in each iteration is sorting\footnote{Comparison-based sorting algorithms   have the worst-case complexity of $O(n \log{n})$ on $n$ inputs \cite[Part II]{cormen2009introduction}.} which is known to be computationally efficient. 

%\blue{Once the thresholds $\{ w_i^{thr} \}_{i=1}^k$ are computed, the corresponding $\theta_{l,i}^{thr}, l \in \mathcal{C}_i, i=1, 2, \cdots, k $ can be obtained through the inverse mappings $w^{-1}_l(.,i), l \in \mathcal{C}_i, i=1, 2, \cdots, k$ given the definition of $\theta_{l,i}^{thr}$ in equation \eqref{valThrsh}.  For each consumer $l$ in class $\mathcal{C}_i, i=1, 2, \cdots, k$, the inverse mapping $w^{-1}_l(\cdot,i): [w_{l,i}^{min}, w_{l,i}^{max}] \longrightarrow \Theta_l$ where, $w_{l,i}^{min} = w_l(\theta_l^{min},i)$ and $w_{l,i}^{max} = w_l(\theta_l^{max},i)$, can be pre-computed with proper discretization of the sets $[w_{l,i}^{min}, w_{l,i}^{max}], \Theta_i$ and stored in a lookup table to be used when needed. The optimal allocation and payment rules are completely determined  in terms of the thresholds $\theta_{l,i}^{thr}, l \in \mathcal{C}_i, i=1, 2, \cdots, k $.}

The allocation and payment procedure does require computing virtual valuations from the reported types and using the inverse mapping\footnote{$w_{l,i}^{min} = w_l(\theta_l^{min},i)$ and $w_{l,i}^{max} = w_l(\theta_l^{max},i)$.} $w^{-1}_l(\cdot,i): [w_{l,i}^{min}, w_{l,i}^{max}] \longrightarrow \Theta_l$ to find the thresholds $\theta_{l,i}^{thr}$ in \eqref{valThrsh}. However, we believe these mappings can be pre-computed with appropriate discretization  and stored in a lookup table to be used when needed. The necessity of computing virtual valuations from  types and vice versa is  a common feature of many mechanism design problems and not unique to our auction.

% that would entail a simple comparison to check whether  for some consumer $l$ in class $\mathcal{C}_i$, the criterion $\theta_l > \theta_{l,i}^{thr}$ is met to finalize the allocation and payment decisions for that consumer.}

%\begin{remark}\label{remark:bKnown}
%Suppose that $\theta_l$ and $b_l$ are independent random variables for all $l \in \mathcal{N}$. In this case the virtual valuation for customer $l$ will take the following form
% \begin{equation}\label{vv_indep}
%w_l(\theta_l) = \theta_l - \frac{1 - F_{l}(\theta_l)}{f_{l}(\theta_l)} \; , \; \forall l \in \mathcal{N},
%\end{equation}
%If we further assume that $\theta_l$ is distributed over the set $[\theta^{min}, \theta^{max}]$ according to the same probability density function $f$  for all $l \in \mathcal{N}$, then the reserve price is the same for all customers and the thresholds $\theta^{thr}_{l,i}$ in \eqref{valThrsh} do  not depend on $l$.  The allocation and payment function can be simplified as follows: For class $\mathcal{C}_i$ define $d_i \coloneqq \max{\{\theta^{res}, \theta_i^{thr}, \theta_{i+1}^{thr}, \cdots, \theta_k^{thr} \}}$.  Customer $l$ in class $\mathcal{C}_i$ gets a good if $\theta_l > d_i$ and its payment simplifies to: $d_i \mathds{1}_{\{\theta_l > d_i\}}$. It is evident in this case that more flexible customers pay less for the good than less flexible customers, that is, 
%\begin{equation}\label{di}
%d_k \le d_{k-1} \le \cdots \le d_1.
%\end{equation}
%\end{remark}

\begin{remark}\label{remark:bKnown}
Suppose that $\theta_l$ and $b_l$ are independent random variables for all $l \in \mathcal{N}$. In this case the virtual valuation for customer $l$ will take the following form
 \begin{equation}\label{vv_indep}
w_l(\theta_l) = \theta_l - \frac{1 - F_{l}(\theta_l)}{f_{l}(\theta_l)} \; , \; \forall l \in \mathcal{N},
\end{equation}
If we further assume that for all $l \in \mathcal{N}$, $\theta_l$ is distributed over the set $[\theta^{min}, \theta^{max}]$ according to the same probability density function $f$, then  the thresholds $\theta^{thr}_{l,i}$ in \eqref{valThrsh} do  not depend on $l$: 
\vspace{-.7em}
\begin{equation}\label{valThrsh2}
 \theta_{i}^{thr} =\Big\{ x  : w(x) = \max\{0,w_i^{thr}, w_{i+1}^{thr}, \cdots, w_k^{thr}\} \Big\}.
\end{equation}
 Moreover, we have $\theta^{thr}_{1} \geq \theta^{thr}_{2} \cdots \geq \theta^{thr}_{k}$.  The allocation and payment functions can be simplified as follows:   Customer $l$ in class $\mathcal{C}_i$ gets a good if $\theta_l > \theta^{thr}_i$ and its payment simplifies to: $\theta^{thr}_i \mathds{1}_{\{\theta_l > \theta^{thr}_i\}}$. It is evident in this case that more flexible customers pay less for the good than less flexible customers.
%, that is, 
%\begin{equation}\label{di}
%d_k \le d_{k-1} \le \cdots \le d_1.
%\end{equation}
\end{remark}

\begin{remark}\label{remark:bdeg}
Suppose the probability distributions for customers' types are such that the flexibility levels are degenerate random variables. This essentially implies that the customers' flexibility sets are common knowledge.  If we further assume that customers' valuations given their flexibility level are identically distributed, then the same observation as in Remark \mbox{\ref{remark:bKnown}} follows: 
Customer $l$ in class $\mathcal{C}_i$ gets a good if $\theta_l > \theta^{thr}_i$ and its payment simplifies to: $\theta^{thr}_i \mathds{1}_{\{\theta_l > \theta^{thr}_i\}}$.
\end{remark}

\section{Conclusion}\label{sec:recap}
%\textcolor{blue}{TO Be Edited}
We studied the problem of designing revenue-maximizing auctions for allocating multiple goods to flexible customers. In our model, each customer is interested in a subset of goods known as its flexibility set and wants to consume one good from this set. A customer's flexibility set and its utility from consuming a good from its flexibility set are its private information.  We characterized the allocation rule for an incentive compatible, individually rational and revenue-maximizing auction in terms of  solutions to  integer programs. The corresponding payment rule was described by an integral equation. We then leveraged the nestedness of flexibility sets to simplify the optimal auction and provided a complete characterization of allocations and payments in terms of simple thresholds.
%  under fixed and replenishable inventory models

A possible extension of  our framework is the case where customers may  demand more than one good from their flexibility sets.
It would also be interesting to study this auction problem under dynamic settings where the set of customers and/or goods can change over time. In such a setting customers may have richer private information that includes their valuation, flexibility sets as well as their temporal presence information.  Moreover, dynamic models can  incorporate supply uncertainties to capture scenarios where the seller relies on uncertain and time-varying resources (such as renewable energy) to serve its customers. 
The auction mechanism then needs to make sequential decisions based on information revealed at or before the current time. 
Investigating these dynamic mechanism design problems will be a key task for  future research.
\vspace{-.5em}
%The model we considered needs to  be further extended to analyze the more general problem where customers' flexibility sets are part of their private information. In this case, customers have two-dimensional types (that is, valuation and flexibility set) which can be strategically misreported in multiple ways.  This complicates the design of an incentive compatible mechanism. Another future direction is to allow for online mechanisms where the set of customers and/or goods can change over time.
% if have a single appendix:
%\appendix[Proof of the Zonklar Equations]
% or
%\appendix  % for no appendix heading
% do not use \section anymore after \appendix, only \section*
% is possibly needed

% use appendices with more than one appendix
% then use \section to start each appendix
% you must declare a \section before using any
% \subsection or using \label (\appendices by itself
% starts a section numbered zero.)
%
\appendices
\section{Proof of Lemma 1}\label{sec:BICLemPrf}
Clearly \eqref{BIC_i_QT} implies  \eqref{BICTheta_i_QT} and \eqref{BICb_i_QT}. To prove the converse, consider flexibility levels $c_i$ and $b_i$ with $c_i \leq b_i$. From \eqref{BICb_i_QT}, we have
%From the BIC constraint for misreports of flexibility level we have:
\begin{equation} \label{L_BIC_s1}
\begin{split}
&\theta_i \; a(b_i) \; Q^T_i(\theta_i, b_i) - T_i(\theta_i, b_i)  \\
&\ge  \theta_i \; a(b_i) \; Q^T_i(\theta_i, c_i) - T_i(\theta_i, c_i) .
\end{split}
\end{equation}
Consider $\theta_i, r_i \in \Theta_i$. From \eqref{BICTheta_i_QT} we have
\begin{equation} \label{L_BIC_s2}
\begin{split}
&\theta_i \; a(c_i) \; Q^T_i(\theta_i, c_i) - T_i(\theta_i, c_i)   \\
&\ge \theta_i \; a(c_i) \; Q^T_i(r_i, c_i) - T_i(r_i, c_i).
\end{split}
\end{equation}
Adding the inequalities in \eqref{L_BIC_s1} and \eqref{L_BIC_s2} we obtain:
\begin{equation} \label{L_BIC_s3}
\begin{split}
&\theta_i \; a(b_i) \; Q^T_i(\theta_i, b_i) - T_i(\theta_i, b_i) + \theta_i \; a(c_i) \; Q^T_i(\theta_i, c_i) \ge \\
&\theta_i \; a(c_i) \; Q^T_i(r_i, c_i) - T_i(r_i, c_i) + \theta_i \; a(b_i) \; Q^T_i(\theta_i, c_i).
\end{split}
\end{equation}
Because of Assumption \ref{assum:Nonz} we have $a(c_i) Q^T_i(\theta_i, c_i) = a(b_i) Q^T_i(\theta_i, c_i)$ and  $a(c_i) Q^T_i(r_i, c_i) = a(b_i) \; Q^T_i(r_i, c_i)$. \eqref{L_BIC_s3} can then be written as
\begin{equation} \label{BIC_gen}
\begin{split}
&\theta_i \; a(b_i) \; Q^T_i(\theta_i, b_i) - T_i(\theta_i, b_i) \ge \\
&\theta_i \; a(b_i) \; Q^T_i(r_i, c_i) - T_i(r_i, c_i),
\end{split}
\end{equation}
which is the two-dimensional BIC constraint of \eqref{BIC_i_QT}. This concludes the proof.
\vspace{-1em}
\section{Proof of Lemma \ref{thm:one}}\label{sec:thmBICValPrf}
\textit{Sufficiency:} 
Suppose $a(b_i) Q^T(r_i, b_i)$ is non-decreasing in $r_i$ and customer $i$'s expected payment is of the form given in \eqref{Thrm1}. Suppose customer $i$'s true type is  $(\theta_i, b_i)$ and it reports $(r_i, b_i)$. Its expected utility is:
\begin{equation}\label{Ui}
U_i(\theta_i, r_i, b_i, b_i) = \theta_i a(b_i) Q_i^{T}(r_i, b_i) - T_i(r_i, b_i).
\end{equation}
We can then use \eqref{Thrm1} to rewrite customer $i$'s expected utility as:
\begin{equation} \label{Ui_re}
\begin{split}
\hspace{-.3cm}U_i(\theta_i, r_i, b_i, b_i) &= (\theta_i - r_i) a(b_i) Q_i^{T}(r_i, b_i) \\
&+ a(b_i) \int\limits_{\theta_i^{min}}^{r_i} Q_i^{T}(s, b_i) ds - K_i(b_i). 
\end{split}
%\raisetag{3\baselineskip}
\end{equation}
%It is easy to verify that the above expression is non-negative for both when $\theta_i \ge r_i$ and $\theta_i < r_i$ because of $a(b_i) Q_i(., b_i)$ being non-decreasing as well as $-K_i(b_i) \ge 0$ (see \eqref{Ki}). This proves individual rationality of the mechanism for user $i$. 

We now need to show that $U_i(\theta_i, \theta_i, b_i, b_i) \ge U_i(\theta_i, r_i, b_i, b_i), \theta_i, r_i \in \Theta_i, b_i \in \{1, 2, \cdots, k\},$ to conclude Bayesian incentive compatibility in valuation for customer $i$. We use the form given in \eqref{Thrm1} to write
\begin{align} \label{IC_pr}
%\begin{split}
&U_i(\theta_i, \theta_i, b_i, b_i) - U_i(\theta_i, r_i, b_i, b_i)  \notag \\
&= a(b_i) \int\limits_{\theta_i^{min}}^{\theta_i} Q^T_i(s, b_i)ds - a(b_i) \int\limits_{\theta_i^{min}}^{r_i} Q^T_i(s, b_i)ds \notag \\
&+(r_i - \theta_i) a(b_i) Q_i^{T}(r_i, b_i)  \notag \\
&= (r_i - \theta_i) a(b_i) Q_i^{T}(r_i, b_i) + a(b_i) \int\limits_{r_i}^{\theta_i} Q^T_i(s, b_i) ds \notag \\
&= \int\limits_{r_i}^{\theta_i}  a(b_i)\{Q^T_i(s, b_i) - Q^T_i(r_i, b_i)\}ds.
%\end{split}
\end{align}
%Because of Assumption \ref{assum:Nonz}, the first term in above expression equals 0; thus we will have
%\begin{equation}\label{IC_pr1}
%\begin{split}
%= a(c_i) \int\limits_{\min{\{\theta_i,r_i\}}}^{\max{\{\theta_i,r_i\}}} Q^T_i(s, c_i) ds - (\theta_i a(d) - r_i a(c_i)) Q_i^T(r_i, c_i)
%\end{split}
%\end{equation}
It is straightforward to verify that because of $a(b_i) Q^T_i(r_i, b_i)$ being non-decreasing in $r_i$, the expression in \eqref{IC_pr} is non-negative for both $r_i < \theta_i$ and $r_i > \theta_i$. Hence
\begin{equation} \label{VP}
U_i(\theta_i, \theta_i,  b_i, b_i) \ge U_i(\theta_i, r_i, b_i, b_i) ~\mbox{for}~ \theta_i, r_i \in \Theta_i,
\end{equation}
which establishes Bayesian incentive compatibility of the mechanism $(q,t)$ in valuation for customer $i$. \\

\textit{Necessity:} 
Suppose $(q,t)$ is  Bayesian incentive compatible in valuation. Consider two candidate valuations $x,y \in \Theta_i, x < y$ that customer $i$ might have. First, assume $(x,b_i)$ is customer $i$'s true type. Then BIC in valuation implies
\begin{equation} \label{Nec_BIC1}
x a(b_i)Q^T_i(x, b_i) - T_i(x, b_i) \ge x a(b_i) Q^T_i(y, b_i) - T_i(y, b_i).
\end{equation}
Now, consider $(y,b_i)$ to be the true type. BIC in valuation gives
\begin{equation} \label{Nec_BIC2}
y a(b_i) Q^T_i(y, b_i) - T_i(y, b_i) \ge y a(b_i) Q^T_i(x, b_i) - T_i(x, b_i).
\end{equation}
Adding \eqref{Nec_BIC1} and \eqref{Nec_BIC2} and simplifying gives
%\begin{equation} \label{Nec_BIC3}
%a(d) Q^T_i(y, c_i) (y-x) \ge a(d) Q^T_i(x, c_i) (y-x)
%\end{equation}
%Since $x < y$, \eqref{Nec_BIC3} implies
\begin{equation} \label{Nec_BIC}
a(b_i) Q^T_i(y, b_i) \ge a(b_i) Q^T_i(x, b_i).
\end{equation}
Therefore, $a(b_i) Q^T_i(r_i, b_i)$ is non-decreasing in $r_i$. 

Let us define $V_i(\theta_i, b_i)$ as customer $i$'s expected utility when its valuation is $\theta_i$ and its flexibility level is $b_i$ and it adopts truth-telling strategy:
\begin{equation} \label{V}
V_i(\theta_i, b_i) \coloneqq U_i(\theta_i, \theta_i, b_i, b_i).
\end{equation}
Using Bayesian incentive compatibility  in valuation \eqref{V} can  be written as
\begin{equation} \label{Vdef}
\begin{split}
V_i(\theta_i, b_i) &= \max\limits_{r_i \in \Theta_i} U_i(\theta_i, r_i, b_i, b_i)\\
&=\max\limits_{r_i \in \Theta_i} \theta_i a(b_i) Q^T_i(r_i, b_i) - T_i(r_i, b_i).
\end{split}
\end{equation}

Using the  integral form of the Envelope Theorem as stated in Theorem 3.1 in \mbox{\cite[Chapter 3]{milgrom2004putting}} and  \mbox{\eqref{Vdef}} it follows that $V_i(\theta_i, b_i)$ satisfies the following equation:
\begin{equation} \label{Vi}
\begin{split}
V_i(\theta_i, b_i) = V_i(\theta_i^{min},b_i) + \int\limits_{\theta_i^{min}}^{\theta_i} a(b_i) Q_i^T(s, b_i) ds.
\end{split}
\end{equation}

Using  \mbox{\eqref{V}} and \mbox{\eqref{Ui}} in \eqref{Vi}, it then follows that $T_i(\theta_i, b_i)$ satisfies the following equation:
\begin{equation} \label{TiForm_Nec}
\begin{split}
T_i(\theta_i, b_i) &= T_i(\theta_i^{min}, b_i) - \theta_i^{min} a(b_i) Q^T_i(\theta_i^{min}, b_i) \\
&+ \theta_i a(b_i) Q^T_i(\theta_i, b_i) - a(b_i) \int\limits_{\theta_i^{min}}^{\theta_i} Q^T_i(s, b_i) ds.
\end{split}
\end{equation}

\eqref{TiForm_Nec} establishes \eqref{Thrm1} with $K_i(b_i) = T_i(\theta_i^{min}, b_i) - \theta_i^{min} a(b_i) Q^T_i(\theta_i^{min}, b_i)$.

\vspace{-.9em}
\section{Proof of Lemma \ref{lem:BICSuffb}}\label{sec:SuffbPrf}
Suppose $(q, t)$ is individually rational and Bayesian incentive compatible in valuation and satisfies  conditions (i) and (ii) of Lemma \ref{lem:BICSuffb}. For a customer of true type $(\theta_i, b_i)$ who reports $(\theta_i, c_i), c_i \le b_i$ consider
\begin{equation}\label{BICbL1}
\theta_i a(b_i) Q_i^T(\theta_i, b_i) - T_i(\theta_i, b_i) - \left(\theta_i a(b_i) Q_i^T(\theta_i, c_i) - T_i(\theta_i, c_i)\right).
\end{equation}
Using \eqref{Thrm1} from Lemma \ref{thm:one} and the second condition of Lemma \ref{lem:BICSuffb} for the two $T_i(\cdot, \cdot)$ terms in \eqref{BICbL1}, we obtain:
\begin{equation} \label{BICb_L2_s4}
\begin{split}
&\int\limits_{\theta_i^{\text{min}}}^{\theta_i} (a(b_i) \; Q_i^T(s,b_i) - a(c_i) \; Q_i^T(s,c_i)) ds \\
&+ \theta_i^{\text{min}} (a(b_i) \; Q_i^T(\theta_i^{\text{min}},b_i) - a(c_i) \; Q_i^T(\theta_i^{\text{min}},c_i)).
\end{split}
\end{equation}
Since $a(b_i)Q^T_i(r_i,b_i)$ is assumed to be non-decreasing in $b_i$, the integral term as well as the term $(a(b_i) \; Q_i^T(\theta_i^{\text{min}},b_i) - a(c_i) \; Q_i^T(\theta_i^{\text{min}},c_i))$ are non-negative. Thus,  the expression in \eqref{BICbL1} is non-negative and hence the BIC constraint in flexibility level (equation \eqref{BICb_i_QT}) is satisfied. 
\vspace{-1em}
\section{Proof of Lemma \ref{lemma:revenue}} \label{sec:revLemmaPf}

The total expected revenue  can be written as 
\begin{align} \label{ExpRev}
\mathbb{E}_{\theta, b}\Big\{\sum\limits_{i=1}^N t_i(\theta, b)\Big\} 
%= \sum\limits_{i=1}^N \mathbb{E}_{\theta}, b\Big\{ t_i(\theta, b) \Big\}
&= \sum\limits_{i=1}^N \mathbb{E}_{\theta_i, b_i}\Big[ \mathbb{E}_{\theta_{-i}, b_{-i}} [t_i(\theta_i, \theta_{-i}, b_i, b_{-i})] \Big] \notag\\
&= \sum\limits_{i=1}^N \mathbb{E}_{\theta_i, b_i}\Big[ T_i(\theta_i, b_i) \Big].
\end{align}
For a mechanism that is individually rational and Bayesian incentive compatible, we can use the result in Lemma \ref{thm:one} to plug in the expression for  $T_i(\theta_i, b_i)$. After some simplifications we obtain that 
%\begin{equation} \label{ExpRev1} 
%\begin{split}
%&\mathbb{E}_{\theta_i, b_i}\Big[ T_i(\theta_i, b_i) \Big] \\
%&=\mathbb{E}_{\theta_i, b_i}\Big[K_i(b_i) + \theta_i \; a(b_i) \; Q^T_i(\theta_i, b_i) \\
%&- a(b_i) \int\limits_{\theta_i^{min}}^{\theta_i} Q^T_i(s, b_i) \; ds \Big] \\
%& = \sum\limits_{b_i=1}^k \int\limits_{\theta_i^{min}}^{\theta_i^{max}}
%\Big[ K_i(b_i) +  \theta_i \; a(b_i) \; Q^T_i(\theta_i, b_i) - \\
%& a(b_i) \; \int\limits_{\theta_i^{min}}^{\theta_i} Q^T_i(s, b_i) \;ds\Big] \: f_i(\theta_i, b_i)\; d\theta_i
%\end{split}
%\end{equation}
%We substitute for $Q_i(\cdot, \cdot)$ from its definition in equation~\eqref{Qi} in terms of $q_i(\cdot,\cdot)$ and after some simplifications we get: \vspace{-5pt}
\begin{equation} \label{ETi} \small
\begin{split}
&\mathbb{E}_{\theta_i, b_i}\Big[ T_i(\theta_i, b_i) \Big] \\
&= \mathbb{E}_{b_i}\Big[ K_i(b_i) \Big] \\
&+ \sum\limits_{b}\int_{\theta}\Big[ a(b_i) \; q_i^T(\theta, b) \Big( \theta_i - \frac{1 - F_{i}(\theta_i | b_i)}{f_{i}(\theta_i | b_i)}\Big)\Big] f(\theta, b) d\theta.
\end{split}
\end{equation}
%where, $f_{\theta | B}(\theta_i | b_i)$ is the probability density function of customer $i$'s valuation conditioned on its flexibility level $b_i$ and $F_{\theta | B}(\theta_i | b_i)$ is the corresponding cumulative distribution function.
%The term $\Big( \theta_i - \frac{1 - F_{\theta | B}(\theta_i | b_i)}{f_{\theta | B}(\theta_i | b_i)}\Big)$ is referred to as  the customer's \textit{virtual type or virtual valuation} in economics terminology and we denote it by $w_i(\theta_i, b_i)$.
% Let us define 
%\begin{equation} \label{wi}
%w_i(\theta_i) = \Big( \theta_i - \frac{1 - F_i(\theta_i)}{f_i(\theta_i)}\Big), \; \forall i \in \mathcal{N}
%\end{equation}

We can now rewrite the auctioneer's total expected revenue in \eqref{ExpRev} as:
\begin{equation} \label{TotExpRev}
\begin{split}
&\sum\limits_{i=1}^N\mathbb{E}_{\theta_i,b_i}\Big[ T_i(\theta_i, b_i) \Big]  
= \sum\limits_{i=1}^N \mathbb{E}_{b_i}\Big[ K_i(b_i) \Big] \\
&+ \sum\limits_{b} \int_{\theta} \; \sum\limits_{i=1}^N \Big[ a(b_i) \; q_i^T(\theta, b) w_i(\theta_i, b_i)\Big] f(\theta, b) d\theta. 
\end{split}
\end{equation}
The second term on the right hand side in \eqref{TotExpRev} is completely determined by the choice of the allocation rule $q(\cdot, \cdot)$. 
Also, note that  Lemmas \ref{thm:one} and \ref{lemma:ir} imply that $K_i(b_i) = T_i(\theta_i^{\text{min}}, b_i) - \theta^{min}_ia(b_i)Q^T_i(\theta_i^{\text{min}}, b_i)  \le 0$. Therefore, a BIC and IR mechanism $(q,t)$ that maximizes the second term on the right hand side in \eqref{TotExpRev} and ensures that $K_i(b_i)=0$ for all $i$ and $b_i$ would provide the largest expected revenue among all BIC and IR mechanisms. 
\vspace{-1em}
%\blue{It thus follows that}
%\begin{align}\label{eq:revMaxIntegral}
%&\max\limits_{(q,t)} \; \mathbb{E}_{\theta, b}\Big\{\sum\limits_{i=1}^N t_i(\theta, b)\Big\} = \max\limits_{(q,t)} \; \sum\limits_{i=1}^N\mathbb{E}_{\theta_i,b_i}\Big[ T_i(\theta_i, b_i) \Big] \notag \\
%&= \max\limits_{q} \; \sum\limits_{b} \int_{\theta} \; \sum\limits_{i=1}^N \Big[ a(b_i) \; q_i^T(\theta, b) w_i(\theta_i, b_i)\Big] f(\theta, b) d\theta. 
%\raisetag{2.5\baselineskip}
%\end{align}
%This concludes the proof.
\section{Proof of Theorem \ref{thm:2}}\label{sec:Thm2Pf}
We first establish that the mechanism $(q^*,t^*)$ is Bayesian incentive compatible and individually rational. Based on the results of Lemmas \ref{lem:BIC2D} - \ref{lem:BICSuffb}, it is sufficient to show the following: 
\begin{enumerate}[(i)]
\item Customer $i$'s expected payment on reporting $r_i$ and $c_i$, $T^*_i(r_i, c_i)$, satisfies \eqref{Thrm1},
\item $T^*_i(\theta_i^{\text{min}}, c_i) = 0 \; , \; \forall c_i \in \{1, 2, \cdots, k\}$, 
\item The expected allocation, $a(c_i) Q^{*T}_i(r_i, c_i)$, is non-decreasing in $r_i$ and $c_i$.
\end{enumerate}
By taking the expectation of $t^*_i(\theta,b)$ over $(\theta_{-i},b_{-i})$ in \eqref{ti_opt}, it is easily established that the expected payment satisfies \eqref{Thrm1} with $K_i(b_i)=0$. Furthermore,  since Assumption \ref{assum:Negwmin} states that $w_i(\theta_i^{min}, b_i) < 0$,  it follows that  $a(b_i) q_i^{*T}(\theta_i^{min}, \theta_{-i}, b) = 0$. If this were not the case then, $q^*$ could not have achieved the maximum in \eqref{qi_opt}. Evaluating \eqref{ti_opt} at $\theta_i^{min}$ then shows that 
 $t^*_i(\theta_i^{\text{min}}, \theta_{-i}, b) = 0$ which further implies that $T^*_i(\theta_i^{min}, b_i) = 0$.

%
%We need to establish Bayesian incentive compatibility and individual rationality properties as characterized in terms of user $i$'s expected payment ($T_i(\theta)$) and expected resource allocation ($Q_i(\theta)$), in Theorem 1. 
%
%Let us write user $i$'s expected payment using the functional form provided for $t^*_i(\theta)$ in Equation~\eqref{ti_opt}:
%\begin{equation} \label{Thr2_T}
%\begin{split}
%T^*_i(r_i) &= \mathbb{E}_{\theta_{-i}}\Big[ t^*_i(r_i,\theta_{-i})\Big] \\
%&= \theta \int_{\theta_{-i}} \: \psi^*_i(r_i, \theta_{-i}) \: f_{\theta_{-i}}(\theta_{-i}) \: d\theta_{-i} \\
%&-\int_{\theta_{-i}} \int\limits_{\theta^{min}}^{\theta} \: \psi^*_i(s,\theta_{-i}) ds \: f_{-i}(\theta_{-i}) \: d\theta_{-i} \\
%& = \theta Q^*_i(r_i) - \int\limits_{\theta^{min}}^{\theta} \: Q^*_i(s) \: ds
%\end{split}
%\end{equation}
%
%Which verifies that provided $K_i = (T^*_i(\theta^{\text{min}}) - \theta^{\text{min}} Q^*_i(\theta^{\text{min}})) = 0$, $T^*_i$ is of the functional form given in Theorem 1 (Equation~\eqref{Thrm1}). 

In order to establish monotonicity of $a(c_i)Q^{*T}_i(r_i,c_i)$ in $r_i$, it is sufficient to argue that $a(c_i)q^T_i(r_i,\theta_{-i}, c_i, b_{-i})$ is non-decreasing in $r_{i}$. % for all $\theta_{-i} \in \Theta_{-i}$
The proof is similar to the arguments in chapters 2-3 of \cite{borgers2015introduction} and basically follows from the fact virtual type $w_i(r_i,c_i)$ is non-decreasing in $r_i$.

To establish monotonicity of $a(c_i)Q^{*T}_i(r_i,c_i)$ in $c_i$, it suffices to show that for any two candidate flexibility levels $\gamma , \lambda \in \{1, 2, \cdots, k\}, \gamma < \lambda$, we will have 
\begin{equation}\label{eq:toprove}
a(\gamma)q^{*T}_i(\theta, \gamma, b_{-i}) \le a(\lambda)q^{*T}_i(\theta, \lambda, b_{-i}),
\end{equation}
 for all $\theta$ and $b_{-i}$. 
%Let the two allocation matrices $({x_1^{\gamma}}^T, {x_2^{\gamma}}^T, \cdots,{ x_N^{\gamma}}^T)^T$ and $({x_1^{\lambda}}^T, {x_2^{\lambda}}^T, \cdots, {x_N^{\lambda}}^T)^T$ be the solutions to the optimization problem in equation \eqref{qi_opt} corresponding to the flexibility level profiles $(\gamma, b_{-i})$ and $(\lambda, b_{-i})$ respectively; hence,
%\begin{equation} \label{aqDefs}
%a(\gamma) \; q^{*T}_i(\theta, \gamma, b_{-i}) = a(b_i) \; {x_i^{\gamma}}^T \; \; , \; \; q^{*T}_i(\theta, \lambda, b_{-i}) = a(\lambda) \; {x_i^{\lambda}}^T 
%\end{equation}
 %Suppose now that user $i$'s true valuation equals $a$, that is, $\theta=a$. Thus we can write
 
For the type profile $(\theta, \gamma, b_{-i})$, the maximum value of the objective function in \eqref{qi_opt} is $a(\gamma) \: q^{*T}_i(\theta, \gamma, b_{-i}) w_i(\theta_i, \gamma) + \sum\limits_{j\neq i} \: a(b_j) \: q^{*T}_j(\theta, \gamma, b_{-i}) \: w_j(\theta_j, b_j)$. Therefore, we must have
\begin{equation} \label{Monton_b_a}
\begin{split}
&a(\gamma) \: q^{*T}_i(\theta, \gamma, b_{-i}) w_i(\theta_i, \gamma) \\
&+ \sum\limits_{j\neq i} \: a(b_j) \: q^{*T}_j(\theta, \gamma, b_{-i}) \: w_j(\theta_j, b_j) \\
&\ge a(\gamma) \: q^{*T}_i(\theta, \lambda, b_{-i}) w_i(\theta_i, \gamma) \\
&+ \sum\limits_{j\neq i} \: a(b_j) \: q^{*T}_j(\theta, \lambda, b_{-i}) \: w_j(\theta_j, b_j).
\end{split}
\end{equation}

Similarly, when the type profile is $(\theta, \lambda, b_{-i})$, the maximum value of the objective function in \eqref{qi_opt} is $a(\lambda) \: q^{*T}_i(\theta, \lambda, b_{-i}) w_i(\theta_i, \lambda) + \sum\limits_{j\neq i} \: a(b_j) \: q^{*T}_j(\theta, \lambda, b_{-i}) \: w_j(\theta_j, b_j)$. Therefore, we must have
\begin{equation} \label{Monton_b_b}
\begin{split}
&a(\lambda) \: q^{*T}_i(\theta, \lambda, b_{-i}) w_i(\theta_i, \lambda) \\
&+ \sum\limits_{j\neq i} \: a(b_j) \: q^{*T}_j(\theta, \lambda, b_{-i}) \: w_j(\theta_j, b_j) \\
&\ge a(\lambda) \: q^{*T}_i(\theta, \gamma, b_{-i}) w_i(\theta_i, \lambda) \\
&+ \sum\limits_{j\neq i} \: a(b_j) \: q^{*T}_j(\theta, \gamma, b_{-i}) \: w_j(\theta_j, b_j).
\end{split}
\end{equation}

Now, adding the two inequalities~\eqref{Monton_b_a}-\eqref{Monton_b_b} gives
\begin{equation} \label{Add_aAndb}
\begin{split}
&( w_i(\theta_i, \lambda)a(\lambda)  -   w_i(\theta_i, \gamma)a(\gamma))\: q^{*T}_i(\theta, \lambda, b_{-i})\\
&\ge (  w_i(\theta_i, \lambda)a(\lambda) -   w_i(\theta_i, \gamma)a(\gamma))\: q^{*T}_i(\theta, \gamma, b_{-i}).
\end{split}
\end{equation}
Define:
\begin{equation}\label{z}
\begin{split}
z_1 &\coloneqq ( w_i(\theta_i, \lambda)a(\lambda)  -   w_i(\theta_i, \gamma)a(\gamma))\: q^{*T}_i(\theta, \lambda, b_{-i}), \\
z_2 &\coloneqq (  w_i(\theta_i, \lambda)a(\lambda) -   w_i(\theta_i, \gamma)a(\gamma))\: q^{*T}_i(\theta, \gamma, b_{-i}).
\end{split}
\end{equation}
\eqref{Add_aAndb} says that 
\begin{equation}\label{Add_aAndb2}
z_1 \geq z_2.
\end{equation}

Let us denote $\omega_{\lambda} \coloneqq w_i(\theta_i, \lambda)$ and $\omega_{\gamma} \coloneqq w_i(\theta_i, \gamma)$.  From the generalized monotone hazard rate condition (see \eqref{MHRC_2}), we know that $\omega_{\lambda} > \omega_{\gamma}$.
  From the definition of vector $a(\cdot)$ (see \eqref{ai}) and $q^*$ (see Assumption \ref{assum:Nonz}) and the fact that $\gamma < \lambda$,  it is easy to see that  $a(\lambda) q^{*T}_i(\theta, \lambda, b_{-i}) \geq  a(\gamma) q^{*T}_i(\theta, \lambda, b_{-i})$ and $a(\lambda) q^{*T}_i(\theta, \gamma, b_{-i}) =  a(\gamma) q^{*T}_i(\theta, \gamma, b_{-i})$.

Depending on the values of $a(\lambda) q^{*T}_i(\theta, \lambda, b_{-i})$,  $a(\gamma) q^{*T}_i(\theta, \lambda, b_{-i})$ and $a(\gamma) q^{*T}_i(\theta, \gamma, b_{-i})$, 
$z_1$ and $z_2$ can take the following values:
\begin{equation} \label{z1}
\begin{split}
z_1= \left\{
    \begin{array}{ll}
        0 & \text{if} \minitab a(\lambda) q^{*T}_i(\theta, \lambda, b_{-i}) = 0 \\
        \omega_{\lambda} - \omega_{\gamma} & \text{if} \minitab a(\gamma) q^{*T}_i(\theta, \lambda, b_{-i}) = 1 \\
        \omega_{\lambda} & \text{if} \minitab (a(\lambda) - a(\gamma)) \: q^{*T}_i(\theta, \lambda, b_{-i}) = 1
    \end{array}
\right.,
\end{split}
\end{equation} 
\begin{equation} \label{z2}
\begin{split}
z_2= \left\{
    \begin{array}{ll}
        0 & \text{if} \minitab a(\gamma) q^{*T}_i(\theta, \gamma, b_{-i}) = 0 \\
        \omega_{\lambda} - \omega_{\gamma} &  \text{if} \minitab a(\gamma) q^{*T}_i(\theta, \gamma, b_{-i}) = 1 
    \end{array}
\right..
\end{split}
\end{equation} 
%Note that in \eqref{z1}, $a(\gamma) q^{*T}_i(\theta, \lambda, b_{-i}) = 1$ and $(a(\lambda) - a(\gamma)) \: q^{*T}_i(\theta, \lambda, b_{-i}) = 1$ imply that $a(\lambda) q^{*T}_i(\theta, \lambda, b_{-i}) = 1$. We know from \eqref{Add_aAndb} that $z_1 \ge z_2$; thus, from the definitions in \eqref{z1} and \eqref{z2} it can be inferred that for $\gamma \le \lambda \; , \; \gamma , \lambda \in \{1, 2, \cdots, k\}$ we will have $a(\gamma)q^{*T}_i(\theta, \gamma, b_{-i}) \le a(\lambda)q^{*T}_i(\theta, \lambda, b_{-i})$. 
%We can now argue that \[a(\gamma) q^{*T}_i(\theta, \gamma, b_{-i}) \le a(\lambda) q^{*T}_i(\theta, \lambda, b_{-i}).\] 
We can establish \eqref{eq:toprove} as follows: The quantities on the left and right hand sides in \eqref{eq:toprove} are either $0$ or $1$.
If $a(\gamma) q^{*T}_i(\theta, \gamma, b_{-i}) = 0$, then \eqref{eq:toprove} is trivially true. It remains to be shown that when $a(\gamma) q^{*T}_i(\theta, \gamma, b_{-i}) = 1$ we also have $a(\lambda) q^{*T}_i(\theta, \lambda, b_{-i}) = 1$. Suppose $a(\gamma) q^{*T}_i(\theta, \gamma, b_{-i}) = 1$   but $a(\lambda) q^{*T}_i(\theta, \lambda, b_{-i}) = 0$.  This  would imply that  $z_2 = \omega_{\lambda} - \omega_{\gamma}$ (which is positive) and $z_1 = 0$; but then $z_1 < z_2$ which is a contradiction of \eqref{Add_aAndb2}. This proves \eqref{eq:toprove}.

%\textcolor{blue}{The argument requires that virtual type is strictly increasing in flexibility but the hazard rate assumption only provides for non-decreasing behavior in flexibility!!}
% thus, $a(\gamma) q^{*T}_i(\theta, \gamma, b_{-i}) \le a(\lambda) q^{*T}_i(\theta, \lambda, b_{-i}) \; , \; \forall \gamma < \lambda , \gamma, \lambda \in \{1, 2, \cdots, k\}$. 

%When $a(\gamma) q^{*T}_i(\theta, \gamma, b_{-i}) = 1$, from \eqref{z2} we have $z_2 = \omega_{\lambda} - \omega_{\gamma}$. \eqref{Add_aAndb} implies that $z_1 \ge z_2$ and thus in this case $z_1$ is equal to either $\omega_{\lambda} - \omega_{\gamma}$ or $\omega_{\lambda}$. From \eqref{z1} we can see that $z_1 = \omega_{\lambda} - \omega_{\gamma}$ if $a(\gamma) q^{*T}_i(\theta, \lambda, b_{-i}) = 1$ and $z_1 = \omega_{\lambda}$ if $(a(\lambda) - a(\gamma)) \: q^{*T}_i(\theta, \lambda, b_{-i}) = 1$. But Assumption \ref{assum:Nonz} implies that $a(\gamma) q^{*T}_i(\theta, \gamma, b_{-i}) = a(\lambda) q^{*T}_i(\theta, \gamma, b_{-i})$ and that $a(\lambda) q^{*T}_i(\theta, \lambda, b_{-i}) \ge a(\gamma) q^{*T}_i(\theta, \lambda, b_{-i})$. These imply that when $a(\gamma) q^{*T}_i(\theta, \gamma, b_{-i}) = 1$ we must have $a(\lambda) q^{*T}_i(\theta, \lambda, b_{-i}) = 1$; thus, $a(\gamma) q^{*T}_i(\theta, \gamma, b_{-i}) \le a(\lambda) q^{*T}_i(\theta, \lambda, b_{-i}) \; , \; \forall \gamma < \lambda , \gamma, \lambda \in \{1, 2, \cdots, k\}$. 

Finally, it is straightforward to see that the allocation rule $q^*(\theta, b)$ which is defined in \eqref{qi_opt} as the maximizer of the weighted sum $\sum\limits_{i=1}^N \: a(b_i) q_i^T(\theta, b) w_i(\theta_i, b_i) $, will naturally maximize the second term on the right hand side of  \eqref{TotExpRev}. Moreover, as argued above, $K_i(b_i)=0$ for all $i$ and $b_i$ under $(q^*,t^*)$.  %Thus, $(q^*,t^*)$  maximizes the auctioneer's total expected revenue. 
Hence, the mechanism $(q^*,t^*)$ is a revenue-maximizing  Bayesian incentive compatible and individually rational mechanism.

\vspace{-1em}
\section{Proof of Lemma \ref{lem:MinRemov}} \label{sec:MinRemovPf}
%At flexibility level 1, the procedure described in Equation \eqref{r_jOpt} removes \textit{at least} $r_1^*$ users from $\mathcal{C}_1$ and hence
%\begin{equation} \label{Proof_i_L1}
%\tilde{n}_1 \le n_1 - r_1^* 
%\end{equation}
%
%Also based on the optimization constraints in \eqref{r_jOpt} we have
%\begin{equation} \label{Proof_i_L2}
%n_1 - r_1^* \le m_1
%\end{equation}
%
%Inequalities in \eqref{Proof_i_L1} and \eqref{Proof_i_L2} imply:
%\begin{equation} \label{Proof_i_L3}
%\tilde{n}_1 \le m_1
%\end{equation}
%
%Similarly the procedure removes \textit{at least} $r_1* + r_2^* + \cdots + r_i^*$ consumers from the set $\mathcal{C}_1 \cup \mathcal{C}_2 \cup \cdots \cup \mathcal{C}_i$ and hence
%\begin{equation} \label{Proof_i_L4}
%\sum\limits_{j=1}^i \tilde{n}_j \le \sum\limits_{j=1}^i (n_j - r_j^*) 
%\end{equation}
%
%On the other hand the optimization inequality constraints ensure that:
%\begin{equation} \label{Proof_i_L5}
%\sum\limits_{j=1}^i (n_j - r_j^*) \le \sum\limits_{j=1}^i m_j
%\end{equation}
%
%Based on \eqref{Proof_i_L4} and \eqref{Proof_i_L5} we can conclude that:
%\begin{equation}
%\sum\limits_{j=1}^i \tilde{n}_j \le \sum\limits_{j=1}^i m_j \; \; ,
%\; \; i = 1, 2, \cdots, k
%\end{equation}
%
%Therefore $r_i^*$'s satisfy all of the adequacy constraints in \hl{equation} \eqref{AdeqIneqs} and hence their feasibility is established. 
%\textit{Proof of part (i)}:\\
Consider any feasible solution of the optimization problem in \eqref{OptimRemovUsers} denoted as $(\tilde{n}_1, \tilde{n}_2, \cdots, \tilde{n}_k)$. We will now show inductively that:
\vspace{-.6em}
\begin{equation} \label{Proof_ii_L1}
\sum\limits_{j=1}^i (n_j - \tilde{n}_j) \ge \sum\limits_{j=1}^i r_j^* \; , \; \forall \; i = 1, 2, \cdots, k.
\end{equation}

For $i=1$ we have:
\vspace{-.6em}
\begin{equation} \label{Proof_ii_L2}
\tilde{n}_1 \le n_1 \; , \; \tilde{n}_1 \le m_1 \; \implies \; \tilde{n}_1 \le \min\{n_1 , m_1\}.
\end{equation}
From this we can write:
\begin{equation} \label{Proof_ii_L3}
n_1 - \tilde{n}_1 \ge n_1 - \min\{n_1 , m_1\} = (n_1 - m_1)^+ = r_1^*.
\end{equation}

Now suppose the inequality in \eqref{Proof_ii_L1} holds for $i$. We now want to prove it also holds for $i+1$. Let us consider two cases based on the possible values of $r_{i+1}^*$: 1) $r^*_{i+1}=0$ and 2) $r^*_{i+1}>0$. When $r^*_{i+1}=0$, it follows directly from the induction hypothesis for $i$ in \eqref{Proof_ii_L1}  that:
\begin{equation} \label{Proof_ii_L4}
\sum\limits_{j=1}^{i+1} (n_j - \tilde{n}_j) \ge \sum\limits_{j=1}^{i+1} r_j^*.
\end{equation}

%Also for the quantification part, when $r_{i+1}^* = 0$ from the inequality constraints in \eqref{OptimCumulRemov} we will get: \vspace{-9pt}
%\begin{equation} \label{Proof_iii_L5}
%\begin{split}
%\sum\limits_{j=1}^i (n_j - r_j^*) + n_{i+1} &\le \sum\limits_{j=1}^{i+1} m_j \\
%\sum\limits_{j=1}^{i+1} (n_j - m_j) &\le \sum\limits_{j=1}^i r_j^*
%\end{split}
%\end{equation}
%
%From induction hypothesis in \eqref{F_l_d} we can substitute for $\sum\limits_{j=1}^i r_j^*$ in \eqref{Proof_iii_L5} and get: \vspace{-7pt}
%\begin{equation} \label{Proof_iii_L6}
%\begin{split}
%\sum\limits_{j=1}^{i+1} (n_j - m_j) &\le \max\limits_{1 \le l \le i} \; \Big[ \Big(\sum\limits_{j=1}^l (n_j - m_j) \Big)^+ \Big] \\
%\Big(\sum\limits_{j=1}^{i+1} (n_j - m_j)\Big)^+ &\le \max\limits_{1 \le l \le i} \; \Big[ \Big(\sum\limits_{j=1}^l (n_j - m_j) \Big)^+ \Big]
%\end{split}
%\end{equation}
%
%From which we can argue that: \vspace{-8pt}
%\begin{equation} \label{Proof_iii_L7}
%\begin{split}
%\max\limits_{1 \le l \le i+1} \; 
% \Big[ \Big(\sum\limits_{j=1}^l (n_j - m_j) \Big)^+ \Big] &= \max\limits_{1 \le l \le i} \Big[ \Big(\sum\limits_{j=1}^l (n_j - m_j) \Big)^+ \Big] \\
%&= \sum\limits_{j=1}^i r_j^* = \sum\limits_{j=1}^{i+1} r_j^*
%\end{split}
%\end{equation}
%%Note that the last line of \eqref{Proof_iii_L7} was resulted because in this case $r_{i+1}^* = 0$. 
%%which thus completes the proof for this special case when $r_{i+1}^* = 0$. 

Now consider the case when $r^*_{i+1}>0$. In this case, from the optimization constraint in \eqref{r_jOpt} it can be verified that $r_{i+1}^* = n_{i+1} + \sum\limits_{j=1}^{i}(n_j - r_j^*) - \sum\limits_{j=1}^{i+1} m_j$; hence: 
\vspace{-.7em}
\begin{equation} 
\begin{split}
\sum\limits_{j=1}^{i+1}(n_j - r_j^*) = \sum\limits_{j=1}^{i+1} m_j,
\end{split}
\end{equation}
which implies
\vspace{-.8em}
\begin{equation}\label{Proof_ii_L5} 
 \sum\limits_{j=1}^{i+1} r_j^* = 
\sum\limits_{j=1}^{i+1} (n_j - m_j).
\end{equation}
From the optimization constraints in \eqref{OptimRemovUsers} we know that:
\begin{equation} \label{Proof_ii_L6}
\sum\limits_{j=1}^{i+1} \tilde{n}_j \le \sum\limits_{j=1}^{i+1} m_j.
\end{equation}
Combining \eqref{Proof_ii_L5} and \eqref{Proof_ii_L6} we get:
\begin{equation} \label{Proof_ii_L7}
\sum\limits_{j=1}^{i+1} (n_j - \tilde{n}_j) \ge 
\sum\limits_{j=1}^{i+1} (n_j - m_j) = \sum\limits_{j=1}^{i+1} r^*_j.
\end{equation}

Thus the inequality in \eqref{Proof_ii_L1} holds for $i+1$ as well. Therefore by induction we can conclude that: $\sum\limits_{j=1}^l (n_j - \tilde{n}_j) \ge \sum\limits_{j=1}^l r_j^*,$ for $l=1,\ldots,k$. Thus,  at least $\sum_{j=1}^l r^*_j$  customers  must be removed from the first $l$ classes to satisfy the   inequalities in \eqref{AdeqIneqs} of Lemma \ref{lem:Adeq}.
%\begin{equation} \label{Proof_ii_L8}
%\sum\limits_{j=1}^k (n_j - \tilde{n}_j) \ge \sum\limits_{j=1}^k r_j^*.
%\end{equation}

To show that the  $\sum_{j=1}^k r^*_j$ is minimum value of the integer program in \eqref{OptimRemovUsers}, consider the following procedure:
\begin{enumerate}
\item  Let $\mathcal{L}_1 := \mathcal{C}_1$. From $\mathcal{L}_1$, $r^*_1$ customers  are removed. The set of remaining customers in $\mathcal{L}_1$ is denoted by $\mathcal{N}_1$.
\item Proceed iteratively:  For $2 \leq i \leq k$, given the set $\mathcal{N}_{i-1}$, 
define  $\mathcal{L}_i := \mathcal{N}_{i-1} \bigcup \mathcal{C}_i.$
Remove $r^*_i$ customers from $\mathcal{L}_i$. The set of remaining customers in $\mathcal{L}_i$ is now defined as $\mathcal{N}_i$.
%\item After the $k^{th}$ iteration, all customers in $\mathcal{N}_k$ are allocated a good from their respective flexibility sets.
\end{enumerate}

It can be verified that the above procedure removes exactly $\sum\limits_{j=1}^k r_j^*$ customers    and creates a demand profile $\mathbf{\tilde n}$ that meets the adequacy condition $\tilde{\textbf{n}} \prec_w \textbf{m}$. 

\bibliographystyle{IEEEtran}
\bibliography{REF}

% Generated by IEEEtran.bst, version: 1.13 (2008/09/30)
\begin{thebibliography}{10}
\providecommand{\url}[1]{#1}
\csname url@samestyle\endcsname
\providecommand{\newblock}{\relax}
\providecommand{\bibinfo}[2]{#2}
\providecommand{\BIBentrySTDinterwordspacing}{\spaceskip=0pt\relax}
\providecommand{\BIBentryALTinterwordstretchfactor}{4}
\providecommand{\BIBentryALTinterwordspacing}{\spaceskip=\fontdimen2\font plus
\BIBentryALTinterwordstretchfactor\fontdimen3\font minus
  \fontdimen4\font\relax}
\providecommand{\BIBforeignlanguage}[2]{{%
\expandafter\ifx\csname l@#1\endcsname\relax
\typeout{** WARNING: IEEEtran.bst: No hyphenation pattern has been}%
\typeout{** loaded for the language `#1'. Using the pattern for}%
\typeout{** the default language instead.}%
\else
\language=\csname l@#1\endcsname
\fi
#2}}
\providecommand{\BIBdecl}{\relax}
\BIBdecl

\bibitem{7548363}
E.~Bitar and Y.~Xu, ``Deadline differentiated pricing of deferrable electric
  loads,'' \emph{IEEE Transactions on Smart Grid}, vol.~8, no.~1, pp. 13--25,
  Jan 2017.

\bibitem{khaledi2013auction}
M.~Khaledi and A.~A. Abouzeid, ``Auction-based spectrum sharing in cognitive
  radio networks with heterogeneous channels,'' in \emph{Information Theory and
  Applications Workshop (ITA), 2013}.\hskip 1em plus 0.5em minus 0.4em\relax
  IEEE, 2013, pp. 1--8.

\bibitem{sengupta2008designing}
S.~Sengupta and M.~Chatterjee, ``Designing auction mechanisms for dynamic
  spectrum access,'' \emph{Mobile Networks and Applications}, vol.~13, no.~5,
  pp. 498--515, 2008.

\bibitem{zhang2012auction}
Y.~Zhang, D.~Niyato, P.~Wang, and E.~Hossain, ``Auction-based resource
  allocation in cognitive radio systems,'' \emph{IEEE Communications Magazine},
  vol.~50, no.~11, pp. 108--120, 2012.

\bibitem{mangili2016bandwidth}
M.~Mangili, F.~Martignon, S.~Paris, and A.~Capone, ``Bandwidth and cache
  leasing in wireless information centric networks: a game theoretic study,''
  \emph{IEEE Transactions on Vehicular Technology}, vol.~66, no.~99, pp.
  679--695, 2017.

\bibitem{vickrey1961counterspeculation}
W.~Vickrey, ``Counterspeculation, auctions, and competitive sealed tenders,''
  \emph{The Journal of finance}, vol.~16, no.~1, pp. 8--37, 1961.

\bibitem{clarke1971multipart}
E.~H. Clarke, ``Multipart pricing of public goods,'' \emph{Public choice},
  vol.~11, no.~1, pp. 17--33, 1971.

\bibitem{groves1973incentives}
T.~Groves, ``Incentives in teams,'' \emph{Econometrica: Journal of the
  Econometric Society}, pp. 617--631, 1973.

\bibitem{cramton2006combinatorial}
P.~Cramton, Y.~Shoham, and R.~Steinberg, ``Combinatorial auctions,'' 2006.

\bibitem{young2014handbook}
P.~Young and S.~Zamir, \emph{Handbook of Game Theory}.\hskip 1em plus 0.5em
  minus 0.4em\relax Elsevier, 2014.

\bibitem{nisan2007algorithmic}
N.~Nisan, T.~Roughgarden, E.~Tardos, and V.~V. Vazirani, \emph{Algorithmic game
  theory}.\hskip 1em plus 0.5em minus 0.4em\relax Cambridge University Press
  Cambridge, 2007, vol.~1.

\bibitem{jin2015quality}
H.~Jin, L.~Su, D.~Chen, K.~Nahrstedt, and J.~Xu, ``Quality of information aware
  incentive mechanisms for mobile crowd sensing systems,'' in \emph{Proceedings
  of the 16th ACM International Symposium on Mobile Ad Hoc Networking and
  Computing}.\hskip 1em plus 0.5em minus 0.4em\relax ACM, 2015, pp. 167--176.

\bibitem{babaioff2009single}
M.~Babaioff, R.~Lavi, and E.~Pavlov, ``Single-value combinatorial auctions and
  algorithmic implementation in undominated strategies,'' \emph{Journal of the
  ACM (JACM)}, vol.~56, no.~1, p.~4, 2009.

\bibitem{myerson1981optimal}
R.~B. Myerson, ``Optimal auction design,'' \emph{Mathematics of operations
  research}, vol.~6, no.~1, pp. 58--73, 1981.

\bibitem{harris1981theory}
M.~Harris and A.~Raviv, ``A theory of monopoly pricing schemes with demand
  uncertainty,'' \emph{The American Economic Review}, pp. 347--365, 1981.

\bibitem{maskin1989optimal}
E.~Maskin, J.~Riley, and F.~Hahn, ``Optimal multi-unit auctions,'' \emph{The
  economics of missing markets, information, and games}, 1989.

\bibitem{malakhov2009optimal}
A.~Malakhov and R.~V. Vohra, ``An optimal auction for capacity constrained
  bidders: a network perspective,'' \emph{Economic Theory}, vol.~39, no.~1, pp.
  113--128, 2009.

\bibitem{de2003combinatorial}
S.~De~Vries and R.~V. Vohra, ``Combinatorial auctions: A survey,''
  \emph{INFORMS Journal on computing}, vol.~15, no.~3, pp. 284--309, 2003.

\bibitem{armstrong2000optimal}
M.~Armstrong, ``Optimal multi-object auctions,'' \emph{Review of Economic
  Studies}, pp. 455--481, 2000.

\bibitem{avery2000bundling}
C.~Avery and T.~Hendershott, ``Bundling and optimal auctions of multiple
  products,'' \emph{The Review of Economic Studies}, vol.~67, no.~3, pp.
  483--497, 2000.

\bibitem{ledyard2007optimal}
J.~O. Ledyard, ``Optimal combinatoric auctions with single-minded bidders,'' in
  \emph{Proceedings of the 8th ACM conference on Electronic commerce}.\hskip
  1em plus 0.5em minus 0.4em\relax ACM, 2007, pp. 237--242.

\bibitem{abhishek2010revenue}
V.~Abhishek and B.~Hajek, ``Revenue optimal auction for single-minded buyers,''
  in \emph{49th IEEE Conference on Decision and Control (CDC)}.\hskip 1em plus
  0.5em minus 0.4em\relax IEEE, 2010, pp. 1842--1847.

\bibitem{hartline2007profit}
J.~Hartline and A.~Karlin, ``Profit maximization in mechanism design,'' in
  \emph{Algorithmic Game Theory (N. Nisan, T. Roughgarden, E. Tardos, and V. V.
  Vazirani, eds.)}.\hskip 1em plus 0.5em minus 0.4em\relax New York, NY, USA:
  Cambridge University Press, 2007, ch.~13, pp. 331--361.

\bibitem{levin1997optimal}
J.~Levin, ``An optimal auction for complements,'' \emph{Games and Economic
  Behavior}, vol.~18, no.~2, pp. 176--192, 1997.

\bibitem{demange1986multi}
G.~Demange, D.~Gale, and M.~Sotomayor, ``Multi-item auctions,'' \emph{Journal
  of Political Economy}, vol.~94, no.~4, pp. 863--872, 1986.

\bibitem{ashlagi2009ascending}
I.~Ashlagi, M.~Braverman, and A.~Hassidim, ``Ascending unit demand auctions
  with budget limits,'' \emph{Massachusetts Inst. Technol., Cambridge, MA, USA,
  Working Paper}, 2009.

\bibitem{borgers2015introduction}
T.~Borgers, R.~Strausz, and D.~Krahmer, \emph{An introduction to the theory of
  mechanism design}.\hskip 1em plus 0.5em minus 0.4em\relax Oxford University
  Press, USA, 2015.

\bibitem{pai2013optimal}
M.~M. Pai and R.~Vohra, ``Optimal dynamic auctions and simple index rules,''
  \emph{Mathematics of Operations Research}, vol.~38, no.~4, pp. 682--697,
  2013.

\bibitem{shaked2007stochastic}
M.~Shaked and G.~Shanthikumar, \emph{Stochastic orders}.\hskip 1em plus 0.5em
  minus 0.4em\relax Springer Science \& Business Media, 2007.

\bibitem{7098381}
A.~Nayyar, M.~Negrete-Pincetic, K.~Poolla, and P.~Varaiya,
  ``Duration-differentiated energy services with a continuum of loads,''
  \emph{IEEE Transactions on Control of Network Systems}, vol.~3, no.~2, pp.
  182--191, June 2016.

\bibitem{cormen2009introduction}
T.~H. Cormen, C.~E. Leiserson, R.~L. Rivest, and C.~Stein, \emph{Introduction
  to algorithms}, third.~ed.\hskip 1em plus 0.5em minus 0.4em\relax MIT press,
  2009.

\bibitem{milgrom2004putting}
P.~R. Milgrom, \emph{Putting auction theory to work}.\hskip 1em plus 0.5em
  minus 0.4em\relax Cambridge University Press, 2004.

\end{thebibliography}
\end{document}